\documentclass[conference]{IEEEtran}
\IEEEoverridecommandlockouts

\usepackage{graphicx}
\usepackage{algorithm}
\usepackage{stackengine}
\usepackage{cite}
\usepackage{color}
\usepackage{xcolor}
\usepackage{amsmath,amsthm,amssymb,amsfonts,bm}
\usepackage{textcomp}
\usepackage{multirow}
\usepackage{url}
\usepackage{makecell}

\def\BibTeX{{\rm B\kern-.05em{\sc i\kern-.025em b}\kern-.08em
    T\kern-.1667em\lower.7ex\hbox{E}\kern-.125emX}}

\usepackage[noend]{algpseudocode}

\makeatletter
\newcommand*{\rom}[1]{\expandafter\@slowromancap\romannumeral #1@}

\newtheorem{lemma}{Lemma}
\newtheorem{theorem}{Theorem}

\newtheorem{remark}{Remark}

\newtheorem{proposition}{Proposition}

\def \v x{\bm x}

\def \v x{\bm X}

\renewcommand{\v}[1]{\ensuremath{\boldsymbol{#1}}}

\newcommand{\recht}[1]{\operatorname{#1}}

\newcommand{\med}{\mathrm{med}}

\begin{document}

\setlength{\lineskip}{-0.6pt}

\title{When Is Distributed Nonlinear Aggregation Private?
Optimality and Information-Theoretical Bounds
}

\author{Wenrui Yu,~\IEEEmembership{Student Member, IEEE}, Jaron Skovsted Gundersen, Richard Heusdens,~\IEEEmembership{Senior Member, IEEE}, \\
and Qiongxiu Li,~\IEEEmembership{Member, IEEE}
\thanks{W. Yu, J. S. Gundersen and Q. Li are with the Department of Electronic Systems, Aalborg University, Denmark.}
\thanks{R. Heusdens is with Netherlands Defence Academy and Delft University of Technology, the Netherlands.}
\thanks{Emails: wenyu@es.aau.dk; jaron@es.aau.dk; r.heusdens@tudelft.nl; qili@es.aau.dk}
\thanks{This work has been submitted to the IEEE for possible publication. Copyright may be transferred without notice, after which this version may no longer be accessible.}}

\maketitle

\begin{abstract}

Nonlinear aggregation is central to modern distributed systems, yet its privacy behavior is far less understood than that of linear aggregation. Unlike linear aggregation where mature mechanisms can often suppress information leakage, nonlinear operators impose inherent structural limits on what privacy guarantees are theoretically achievable when the aggregate must be computed exactly. This paper develops a unified information-theoretic framework to characterize privacy leakage in distributed nonlinear aggregation under a joint adversary that combines passive (honest-but-curious) corruption and eavesdropping over communication channels. 

We cover two broad classes of nonlinear aggregates: order-based operators (maximum/minimum and top-$K$) and robust aggregation (median/quantiles and trimmed mean). We first derive fundamental lower bounds on leakage that hold without sacrificing accuracy, thereby identifying the minimum unavoidable information revealed by the computation and the transcript. We then propose simple yet effective privacy-preserving distributed algorithms, and show that with appropriate randomized initialization and parameter choices, our proposed approaches can attach the derived optimal bounds for the considered operators. Extensive experiments validate the tightness of the bounds and demonstrate that network topology and key algorithmic parameters (including the stepsize) govern the observed leakage in line with the theoretical analysis, yielding actionable guidelines for privacy-preserving nonlinear aggregation.

\end{abstract}


\section{Introduction}
Distributed data aggregation has become a fundamental paradigm powering modern intelligent systems.
From real-time environmental and health monitoring in sensor networks~\cite{hu2018event,rezaeibagha2020secure} to large-scale machine learning tasks involving decentralized datasets~\cite{pillutla2022robust}, diverse sources continuously generate data that must be integrated to extract global patterns and support decision-making. While distributed aggregation enables efficiency, robustness, and scalability, it also raises a critical concern: aggregated information may inadvertently reveal sensitive user data, especially in the presence of adversaries. Such privacy leakage can compromise user trust and create regulatory risks, making privacy-preserving aggregation a central challenge in distributed systems.   

Although linear aggregation methods such as sums and averages have been extensively studied, many emerging applications require more expressive nonlinear aggregation. These include order-based operators (e.g., maximum, minimum, top-
$k$) that depend on the relative ordering of inputs, as well as robust statistical aggregators (such as median and trimmed mean) designed to suppress the influence of noise, outliers, or adversarial updates. Such nonlinear operators play a central role in modern distributed systems but introduce privacy challenges fundamentally different from linear aggregation.

Most existing work on privacy-preserving distributed aggregation focuses on \emph{linear} tasks such as average consensus~\cite{mo2016privacy,he2018privacy,wang2019privacy}. Differential-privacy approaches~\cite{huang2015differentially,nozari2018differentially,zhang2016dynamic,zhang2018improving,xiong2020privacy} provide formal guarantees via noise insertion but typically degrade the output accuracy. Other lines of work include subspace-perturbation methods~\cite{Jane2020ICASSP,Jane2020LS,li2020privacy}, and cryptographic approaches based on secure multi-party computation (SMPC)~\cite{gupta2017privacy,li2019privacyA,tjell2020privacy,tjell2019privacy,xu2015secure,li2019privacyS,shoukry2016privacy,zhang2019admm} which largely rely on a linear structure and cannot be straightforwardly adapted to nonlinear aggregation. Only a few studies directly address nonlinear functions; for instance, \cite{ruan2019secure} designs an SMPC protocol for max/min consensus, but it can incur high computational and communication overhead.

We previously conducted two preliminary studies on individual nonlinear operators under specific assumptions: a lossless maximum consensus method formulated via augmented-graph optimization~\cite{yu2025privacy}, and an initial characterization of privacy conditions for exact median consensus~\cite{yu2025optimal}. These studies focused on individual operators and served as early steps toward understanding privacy in nonlinear aggregation.   In this paper, we advance beyond these isolated cases by developing a unified analytical framework that characterizes the privacy behavior of general nonlinear aggregation functions. Our framework clarifies how structural properties of different operator classes determine their fundamental privacy limits and yields the conditions under which accuracy-preserving privacy is attainable.  Our contributions are summarized as:
\begin{itemize}
\item \textbf{Information-theoretic limits:} We derive fundamental lower bounds on privacy leakage for exact nonlinear aggregation and identify when perfect privacy is theoretically achievable.
\item \textbf{Operator-class characterization:} We analyze how the structural properties of order-based and robust-statistical operators determine their intrinsic privacy behavior.
\item \textbf{Algorithmic optimality:} For canonical cases, we establish the conditions under which PDMM-type algorithms can attain the theoretical privacy limits, supported by empirical validation.
\item \textbf{Key factors influencing privacy:} We analyze and verify how network topology, initialization choices, and algorithmic parameters affect achievable privacy, in alignment with our theoretical analysis.
\end{itemize}

\subsection{Outline and notations}

The structure of this paper is as follows. Section~\ref{sec.pre} presents the preliminaries, Section~\ref{sec.ideal} introduces the fundamental lower bound on leakage, Section~\ref{sec.order} analyzes privacy in order-based aggregations, Section~\ref{sec.robust} analyzes privacy in robust statistical aggregations, Section~\ref{sec.experiment} presents experimental validation, and finally, we conclude in Section~\ref{sec.conclusion}.

We use lowercase letters $(x)$ for scalars, lowercase boldface letters $(\boldsymbol{x})$ for vectors, uppercase boldface letters $(\boldsymbol{X})$ for matrices and uppercase letters $(X)$ for random variables whose realizations are scalars denoted by $x$. $|\cdot|$ denotes the cardinality of a set and $(\boldsymbol{X})^{\top}$ denote the transpose of $\boldsymbol{X}$. $||\cdot||_p$ denotes the $L_p$ norm and $||\cdot||$ denote the $L_2$ norm by default.
$(\overline{\cdot})$ denotes the averaging of the elements in the vector.
$(\cdot)^{(t)}$ denotes the value at iteration $t$.

\section{Preliminaries}\label{sec.pre}

\subsection{Distributed optimization}

The network is modeled as an undirected graph $\mathcal{G} = (\mathcal{V}, \mathcal{E})$, where $\mathcal{V} = \{1, \ldots, n\}$ denotes the set of vertices (nodes) representing network participants. The edge set $\mathcal{E} \subseteq \mathcal{V} \times \mathcal{V}$ defines the pairwise communication links between nodes. For each node $i \in \mathcal{V}$, its neighborhood $\mathcal{N}_i = \{j \in \mathcal{V} \mid (i,j) \in \mathcal{E}\}$ consists of adjacent nodes, with $d_i = |\mathcal{N}_i|$ denoting its degree.

The distributed optimization problem is formulated as
\begin{equation}
\begin{aligned}
& \underset{\{\v x_i:{i \in \mathcal{V}}\}}{\text{minimize}} 
& & \sum_{i \in \mathcal{V}} f_i(\v x_i), \\
& \text{subject to}
& & \v A_{ij} \v x_i + \v A_{ji} \v x_j = \v b_{i,j}, \quad \forall (i,j) \in \mathcal{E},
\end{aligned}
\label{eq:problem}
\end{equation}
where $\v x_i \in \mathbb{R}^d$ is the local variable, and $f_i: \mathbb{R}^d \to \mathbb{R}$ is a convex, closed, and proper (CCP) objective function that encapsulates the private data $\v s_i \in \mathbb{R}^d$. $\v A_{ij}, \v A_{ji} \in \mathbb{R}^{m \times d}$ and $\v b_{i,j} \in \mathbb{R}^m$ define coupling constraints between adjacent nodes.

To address the optimization problem in \eqref{eq:problem}, we employ the Primal-Dual Method of Multipliers \cite{zhang2017distributed,sheron2018derivation}. The iterative update equations for a given node $i \in \mathcal{V}$ are formulated as
\begin{align}
    \nonumber \v x_i^{(t)}&=\arg \min _{\v x_i}\bigg(f_i(\v x_i)\\ 
     & +\sum_{j \in \mathcal{N}_i}\bigg(\v z_{i \mid j}^{(t)\top} \v A_{ij} \v x_i+\frac{c}{2}\left\|\v A_{ij} \v x_i-\frac{1}{2} \v b_{i, j}\right\|^2\bigg)\bigg),\label{eq.x_update}\\
    \v z_{i\mid j}^{(t+1)}&=(1-\theta)\v z_{i\mid j}^{(t)}+\theta (\v z_{i \mid j}^{(t)}+2 c\v A_{ij} \v x_i^{(t)}- \v b_{i, j}), \label{eq.z_update} 
\end{align}
where $\v y$ and $\v z$ are auxiliary variables, $\theta \in (0,1)$ is an averaging constant, and $c > 0$ is a constant controlling the convergence rate (a larger $c$ gives a smaller step-size, hence a smaller convergence rate). When the objective function is uniformly convex, the algorithm converges even for $\theta = 1$ (standard PDMM) \cite{heusdens2024distributed}. For non-uniformly convex problems, we primarily analyze the case where $\theta = \frac{1}{2}$. This choice corresponds to the $\frac{1}{2}$-averaged version of PDMM, which is equivalent to ADMM~\cite{sheron2018derivation}.

\subsection{Adversary model}\label{subsec.adv}
We analyze two prevalent adversary models in the context of graph-based systems. The first one is the passive adversary, also referred to as the honest-but-curious model. In this setting, corrupted nodes adhere strictly to the prescribed protocol but may collude with each other, exchanging information to infer sensitive data. 

The second adversary model is the eavesdropping adversary, which possesses the ability to monitor all communications transmitted over unsecured channels. These two adversarial models are assumed to operate in concert, combining internal collusion and intercepted external communications to reconstruct the private data of honest nodes.

Throughout the paper we assume that inter-node messages are transmitted over unsecured channels, so the eavesdropping adversary observes the full communication transcripts. We do not rely on encrypted channels or SMPC-style cryptographic protections, as such mechanisms introduce additional key-management assumptions and may incur substantial computation/communication overhead in iterative protocols.

\subsection{Evaluation metrics}\label{ss.eval}
The performance of the proposed algorithm is rigorously evaluated with respect to two fundamental criteria: \textit{accuracy}, which measures the correctness of the computational results, and \textit{privacy}, which quantifies the protection of sensitive information. 

\noindent\textbf{Output accuracy}: We want to assess the proximity between the optimization outcomes of the privacy-preserving algorithm and those achieved by the original non-privacy-preserving algorithms. The accuracy is measured by the Mean Squared Error (MSE)
\begin{align}\label{eq.mse}
    \frac{1}{|\mathcal{V}|}\sum_{i\in\mathcal{V}}\|\v x_i^{(t_{\recht{max}})}- \v x_i^{*}\|^{2},
\end{align}
where $t_{\recht{max}}$ denotes the maximum number of iterations and $\v x_i^*$ the optimal solution on node $i$ to the optimization problem. For consensus problems, we have $\boldsymbol{x}_i^* = \boldsymbol{x}_j^*$ for all $(i,j) \in \mathcal{E}$.
We therefore denote the common consensus value by $\boldsymbol{x}^*$.

\noindent\textbf{Individual privacy}:
 We adopt mutual information \cite{cover2012elements} as the metric for assessing individual privacy as it is shown to be effective in the literature \cite{duchi2013local,kairouz2014extremal,li2022privacy}. Given the random variable $S_i$ 
representing the private data at node $i$ and $\mathcal{O}$ representing the total information that the adversary can observe, the  mutual information $I(S_i;\mathcal{O})$  measures the amount of information learned about $S_i$ by observing $\mathcal{O}$, which is given by
\begin{equation}
    \nonumber I(S_i;\mathcal{O})=h(S_i)-h(S_i\mid \mathcal{O}),
\end{equation}
where $h(\cdot)$ denotes differential entropy, assuming it exists.
When $I(S_i;\mathcal{O})=h(S_i)$, the adversary has sufficient information to fully infer $\v s_i$. When $I(S_i;\mathcal{O})=0$, the adversary cannot infer any information about  $S_i$ given the available information ${\mathcal O}$. Let $\mathcal{V}_c$ denote the set of corrupted nodes corrupted by a passive adversary and let $\mathcal{V}_h = \mathcal{V}\,\backslash\, \mathcal{V}_c$ represents the set of honest nodes in the network. 
$\mathcal{O}$ consists of the inputs and outputs of $\mathcal{V}_c$ and all the random choices taken at the corrupted nodes. When considering eavesdropping, $\mathcal{O}$ also contains all messages sent between any parties.

\subsection{Nonlinear aggregation formulation}\label{ss.formulation}
We consider distributed aggregation of local values $\{s_i\}_{i\in\mathcal{V}}$ and model the desired output as 
\begin{align}\label{eq:agg_def}
    x^* = f\!\left(\{s_i\}_{i\in\mathcal{V}}\right),
\end{align}
where $f(\cdot)$ can be linear or nonlinear\footnote{For simplicity, we use a scalar representation throughout the coordinate-wise cases in this paper, although the formulation can be straightforwardly extended to the vector case.}.

A typical linear example can be the weighted average, i.e., $x^*=\frac{1}{|\mathcal{V}|}\sum_{i\in\mathcal{V}} w_i s_i,$
while nonlinear examples include 
\begin{align*}\label{eq:nonlinear_examples}
    x^*=\operatorname{max}_{i\in\mathcal{V}}\{s_i\},\qquad
    x^*=\operatorname{med}_{i\in\mathcal{V}}\{s_i\},
\end{align*}
where $\operatorname{med}(\cdot)$ denotes the median.

These (non)linear operators can be formulated as distributed optimization problems. For example, maximum consensus can be formulated by introducing indicator functions\footnote{
Throughout this paper, $\mathcal{I}_A(x)=\infty$ if $x\in A$ and $0$ otherwise,
while $\mathbb{I}_A(x)=1$ if $x\in A$ and $0$ otherwise.
}~\cite{venkategowda2020privacy} 
\begin{equation*}
\begin{array}{ll}
\underset{\{x_i: i \in \mathcal{V}\}}{\text{minimize}} 
& {\displaystyle \sum_{i\in \mathcal V} \left(x_i + \mathcal{I}_{\{x_i \geq s_i\}}(x_i)\right)}, \\[3mm]
\text{subject to} 
& x_i - x_j = 0, \quad (i,j)\in \mathcal{E}.
\end{array}
\label{eq:max_primal}
\end{equation*}
The consensus constraints enforce that all $x_i$ converge to a common value. Under A/PDMM, this corresponds to setting $a_{ij}x_i+a_{ji}x_j=b_{ij}$ where $a_{ij} = -a_{ji} = 1$ for $i<j$ and $b_{ij}=0$.

Similarly, the median consensus problem can also be formulated in a distributed manner as~\cite{deplano2023unified,yu2025optimal} 
\begin{equation*}
\begin{array}{ll}
\underset{\{x_i: i\in\mathcal{V}\}}{\text{minimize}} 
& {\displaystyle \sum_{i\in \mathcal V} \|x_i - s_i\|_1}, \\[3mm]
\text{subject to} 
& x_i - x_j = 0, \quad (i,j)\in \mathcal{E}.
\end{array}
\end{equation*}
Here the objective is the $L_1$ norm, reflecting the deviation from the median operator.

Additional nonlinear aggregation formulations will be introduced in subsequent sections.

\section{Optimal privacy bounds in the ideal world}\label{sec.ideal}

This section derives information-theoretic privacy benchmarks under the requirement of \emph{perfect accuracy}: the distributed protocol outputs the exact nonlinear aggregate, i.e., there is no accuracy loss. In this regime, privacy cannot in general be made arbitrarily strong. Instead, we characterize the \emph{minimum unavoidable leakage} implied by the function itself together with the adversary model.

\subsection{Ideal world in SMPC}
The notion of minimum unavoidable leakage is closely related to the \emph{ideal-world} paradigm in SMPC~\cite{Cramer2015}. In the ideal world, a trusted third party computes the desired aggregate and returns only the output. Against a passive (honest-but-curious) adversary, this means the adversary should learn nothing beyond what is implied by the output and the inputs of corrupted nodes. 

Using the aggregation output $x^*$ defined in \eqref{eq:agg_def}, the classical ideal-world leakage given a set of passive corrupt nodes $\mathcal{V}_c$ is
\begin{align}\label{eq:Classical_ideal}
    I\!\left(S_i;\ X^*,\ \left\{S_j\right\}_{j\in\mathcal{V}_c}\right),
    \quad i\in \mathcal{V}_h .
\end{align}

\subsection{Ideal-world under passive corruption and eavesdropping}
Recall the adversary model in Section~\ref{subsec.adv}, where passive corruption is combined with a global eavesdropper, the adversary may learn additional information that is not implied by $(X^*,\{S_j\}_{j\in\mathcal{V}_c})$. Consequently, when we refer to an ``ideal-world'' privacy bound, we mean the minimum leakage consistent with both passive corruption and transcript eavesdropping.

In the following, we identify this unavoidable additional leakage for two canonical nonlinear aggregation examples: maximum (order-based aggregation) and median (robust aggregation). The same approach can easily extend to other applications such as quantiles.

\subsection{Optimal privacy bound in maximum consensus}\label{subsec.ideal.max}
Consider the maximum consensus $X^*=S_{\max}=\max_{j\in\mathcal{V}}\{S_j\}$. Under global transcript access, any correct distributed computation of $S_{\max}$ necessarily reveals which node attains the maximum.
This information is not determined in general by $(S_{\max},\{S_j\}_{j\in\mathcal{V}_c})$ and therefore must be included in the ideal-world benchmark under our adversary model.

To do so, we define the maximizer membership variables
\[
\mathbb{I}_{\{S_j = S_{\max}\}}(S_j),
\]
where \(\mathbb{I}_{\{S_j = S_{\max}\}}(S_j)= 1\) if node \(j\) attains the maximum, and \(0\) otherwise.

\begin{proposition}[Unavoidable leakage for maximum consensus]\label{prop:max_leak}
    Let $\mathcal{G}=(\mathcal{V},\mathcal{E})$ be an undirected graph. Assume a distributed algorithm is carried out such that each node outputs $s_{\max}=\max_{j\in\mathcal{V}} s_j$ at termination, where node $j$ holds $s_j$. Under the adversary model of Section~\ref{subsec.adv} (passive corruption with $\mathcal{V}_c\neq\emptyset$ and global eavesdropping), the adversary can determine the maximizer membership variables $\{\mathbb{I}_{\{S_j = S_{\max}\}}(S_j)\}_{j\in\mathcal{V}}$ from its observation.

    Consequently, for every honest node $i\in\mathcal{V}_h$, if $\mathcal{O}$ denotes the adversary's total observation, the leakage satisfies the information-theoretic lower bound
    \begin{align}\label{eq:ideal_max_bound}
        I(S_i;\mathcal{O})
        \ \ge\
        I\!\left(S_i;\ S_{\max},\ \left\{S_j\right\}_{j\in\mathcal{V}_c},\ \left\{\mathbb{I}_{\{S_j = S_{\max}\}}(S_j)\right\}_{j\in \mathcal{V}}\right).
    \end{align}
\end{proposition}
The proof is given in Appendix~\ref{app:proofs_max}. Motivated by Proposition~\ref{prop:max_leak}, we use the right-hand side of \eqref{eq:ideal_max_bound} as the ideal-world benchmark for maximum aggregation.

\subsection{Median aggregation}\label{subsec.ideal.med}
Consider $S_{\med}=\mathrm{med}_{j\in\mathcal{V}}\{S_j\}$. Under global transcript access, median computation reveals additional order information beyond $(S_{\med},\{S_j\}_{j\in\mathcal{V}_c})$: besides the identity of the median holder, the adversary learns, for each node, whether its value is smaller than, greater than, or equal to the median. This side information is therefore part of the ideal-world benchmark under our adversary model.

Introduce the membership variables
\[
\Big\{\mathbb{I}_{\{S_j \leq S_{\med}\}}(S_j),\mathbb{I}_{\{S_j \geq S_{\med}\}}(S_j)\}\Big\}_{j\in\mathcal{V}}.
\]
Note that the pair $(\mathbb{I}_{\{S_j \leq S_{\med}\}}(S_j),\mathbb{I}_{\{S_j \geq S_{\med}\}}(S_j))$ distinguishes $S_j<S_{\med}$, $S_j>S_{\med}$, and $S_j=S_{\med}$ (where both indicators equal $1$). 

\begin{proposition}[Unavoidable leakage for median consensus]\label{prop:med_leak}
    Let $\mathcal{G}=(\mathcal{V},\mathcal{E})$ be an undirected graph. Assume a distributed algorithm is carried out such that each node outputs $s_{\med}=\mathrm{med}_{j\in\mathcal{V}} s_j$ at termination, where node $j$ holds $s_j$. Under the adversary model of Section~\ref{subsec.adv} (passive corruption with $\mathcal{V}_c\neq\emptyset$ and global eavesdropping), the adversary can determine, for every $j\in\mathcal{V}$, whether $S_j<S_{\med}$, $S_j>S_{\med}$ or $S_j=S_{\med}$ (captured by $\mathbb{I}_{\{S_j \leq S_{\med}\}}(S_j)$ and $\mathbb{I}_{\{S_j \geq S_{\med}\}}(S_j)$).

    Consequently, for every honest node $i\in\mathcal{V}_h$, the leakage satisfies
    \begin{align}\label{eq:ideal_med_bound}
        I(S_i;\mathcal{O})
        \ \ge\
        I\!\Big(&S_i;\ S_{\med},\ \left\{S_j\right\}_{j\in\mathcal{V}_c},\nonumber\\
        &\left\{\mathbb{I}_{\{S_j \leq S_{\med}\}}(S_j),\mathbb{I}_{\{S_j \geq S_{\med}\}}(S_j)\right\}_{j\in\mathcal{V}}\Big).
    \end{align}
\end{proposition}
\begin{proof}
    See Appendix~\ref{app:proofs_med}.
\end{proof}

Motivated by Proposition~\ref{prop:med_leak}, we use the right-hand side of \eqref{eq:ideal_med_bound} as the ideal-world benchmark for median aggregation under the adversary model of Section~\ref{subsec.adv}.

\subsection{Discussion and extension}
The privacy bounds above will be used as the privacy target in later sections: under perfect accuracy and the adversary model of Section~\ref{subsec.adv}, a protocol is information-theoretically optimal if its leakage matches the corresponding ideal world privacy bounds. The same methodology extends to other order statistics (e.g., quantiles) by incorporating the appropriate rank/side-information indicators implied by transcript access, i.e. if $x^*$ is the output of a distributed quantile computation, the adversary will learn $\mathbb{I}_{\left\{S_j\leq X^*\right\}}$ and $\mathbb{I}_{\left\{S_j\geq X^*\right\}}$ for all nodes $j$. Similarly, we can extend our ideal maximum to top-$K$, since the adversary under our model always will learn which nodes possess the top-$K$ values when a distributed computation of this is carried out. This can be shown by similar arguments as for the maximum case. Assuming that the trimmed mean outputs both the average and top-$K$ maximum and minimum the ideal leakage can be obtained by combining the ideal leakage of these three functionalities.

\section{Order-based aggregation}\label{sec.order}
This section studies \emph{order-based} nonlinear aggregation, where the aggregate is determined by the relative ordering of inputs rather than by linear combinations. We focus on three classical tasks: maximum, minimum, and top-$K$. Although median and quantiles are also order-based, we defer them to Section~\ref{sec.robust} because they serve as core primitives in robust statistical aggregation and their privacy behavior differs qualitatively from extrema.

\subsection{Maximum consensus\label{sec.maximum}}
We begin with distributed maximum consensus. Beyond applications such as leader selection, detection, and monitoring, maximum consensus is also a canonical example for privacy analysis: as shown in the ideal-world discussion in Section~\ref{sec.ideal}, under eavesdropping one cannot in general hide \emph{which node attains the maximum}. The goal here is therefore not to eliminate all leakage, but to design and analyze a protocol whose leakage matches this ideal-world benchmark as closely as possible.

\subsubsection{Proposed approach}
The maximum consensus problem is formulated in Section~\ref{ss.formulation} and solved via a PDMM-type primal-dual iteration. The privacy mechanism is implemented through the \emph{random initialization} of the edge-associated dual variables $\{z_{i|j}^{(0)}\}$, which appear in the left term inside the max operator in \eqref{eq.xi_up}. Intuitively, when the initialization is chosen appropriately, this term can remain larger than $s_i$ for non-maximizer nodes, so that the max operator never selects $s_i$ during the run. 

The update rules for the primal variable $x_i$ and the dual variables $z_{j|i}$ are given by~\cite{deplano2023unified}
\begin{equation}\label{eq.xi_up}x_{  i}^{(t)} =\max\left(\frac{ -1- \sum_{  j \in {\mathcal N}_i}  a_{ij}z_{  i|j}^{(t)}}{cd_i},s_i\right)\end{equation}
                    \begin{equation}\label{eq.zji_up}z_{  j|i}^{(t+1)}
        =\frac{1}{2}z_{j\mid i}^{(t)}+\frac{1}{2}(z_{  i|j}^{(t)}+2c a_{ij}x_i^{(t)})
                \end{equation}
Using these updating functions we present our privacy-preserving maximum consensus algorithm in Algorithm~\ref{alg:max_primal}.
\begin{algorithm}[ht]
  \caption{Privacy-preserving maximum consensus}
  \label{alg:max_primal}
  \begin{algorithmic}
     \ForAll{ $i \in \mathcal{V}, j \in \mathcal{N}_i$,}
     \State Randomly initialize $z_{i\mid j}^{(0)}\sim\mathcal{N}(\pm\mu_z,\sigma_z^2)$ 
     \Comment{Initialization}
     \State $\text{Node}_j \leftarrow \text{Node}_i(z_{  i|j}^{(0)})$
     \EndFor
          \For{$t=0,1,\cdots,t_{\max}-1$} 
            \ForAll{$i \in \mathcal{V}$ }
            \State Update $x_{  i}^{(t)}$ with \eqref{eq.xi_up}
            \State Update $\forall j\in \mathcal{N}_i: z_{  j|i}^{(t+1)}$ with \eqref{eq.zji_up}
                \State $\text{Node}_{j\in \mathcal{N}_i} \leftarrow \text{Node}_i(x_i^{(t)})$\Comment{Broadcast}
                \ForAll{ $j \in \mathcal{N}_i$} 
        \State Update $z_{  j|i}^{(t+1)}$ with \eqref{eq.zji_up}
            \EndFor
      \EndFor
      \EndFor
  \end{algorithmic}
\end{algorithm}

We note that the $\pm$ sign in the initialization of Algorithm~\ref{alg:max_primal} is chosen consistently with the signed edge structure (through $a_{ij}$); the role of $(\mu_z,\sigma_z)$ and how to choose them for privacy-optimal behavior are discussed in Section~\ref{ss.max_sys}.

\subsubsection{Performance Analysis}\label{ss.max_prime_privacy}
We use the privacy metric and adversary model defined in Section~\ref{ss.eval}. In particular, we consider passive corruption of $\mathcal{V}_c$ together with global eavesdropping, and (in this section) plaintext communication. This implies that initialization messages carrying $\{z_{i|j}^{(0)}\}$ are also observable; under this assumption, broadcast and unicast transmissions lead to the same adversarial observation.

For corrupt nodes $j \in \mathcal{V}_c$, the observation over all iterations consist of:
\[
\big\{S_j,\, X_j^{(t)},\, Z_{j \mid k}^{(t)},\, Z_{k \mid j}^{(t)}\big\}_{j \in \mathcal{V}_c,\, k \in \mathcal{N}_j,\, t \in \mathcal{T}},
\]
where $\mathcal{T} =\{0, \dots, t_{\mathrm{max}}\}$ denotes the complete set of iteration indices and $t_{\mathrm{max}}$ represents the final iteration. 
Eavesdropping adversaries intercept the channel-communicated data and, therefore, have knowledge of
\[
\big\{Z_{j \mid k}^{(0)},\, X_j^{(t)}\big\}_{j \in \mathcal{V},\, (j,k) \in \mathcal{E},\, t \in \mathcal{T}}.
\]

After eliminating informational redundancies accessible to both adversarial models, the individual privacy metric reduces to
\begin{align}\label{eq.adv_info_max}
    \nonumber I(S_i;\mathcal{O})
    \overset{}{=}&I(S_i;\{S_j,Z_{j \mid k}^{(t)},\, Z_{k \mid j}^{(t)}\}_{j\in\mathcal{V}_c, k \in \mathcal{N}_j, t \in \mathcal{T}},\\
    &\{X_{j}^{(t)}\}_{j\in \mathcal{V},t\in\mathcal{T}},\{Z_{j\mid k}^{(0)}\}_{(j,k)\in\mathcal{E}}).
\end{align}

The next theorem identifies a trajectory-dependent condition that determines whether the max operator in \eqref{eq.xi_up} ever reveals $s_i$ directly.

\begin{theorem}\label{thm.max_prime}
Let ${\mathcal V}_{p}$ be the set of nodes satisfying
    \begin{equation}\label{eq.condition_P}
    \frac{ -1- \sum_{  j \in {\mathcal N}_i}  a_{  ij}z_{  i|j}^{(t)}}{cd_i}>s_i
  \end{equation}
for all $t\in \mathcal T$, and let $Y_i = \mathbb{I}_{\{i\in\mathcal{V}_p\}}(S_i)$. Then
\begin{align}\label{eq.theo11}
    \nonumber I(S_i;\mathcal{O})=&I(S_i;Y_i)+P(Y_i=0)I(S_i;\mathcal{O}|Y_i=0)\\
    &+P(Y_i=1)I(S_i;\mathcal{O}|Y_i=1)
\end{align}
where
\begin{align}
    \nonumber I(S_i;\mathcal{O}|Y_i=0)=&I(S_i;S_i\mid Y_i=0) \\
    \nonumber I(S_i;\mathcal{O}|Y_i=1)=&I(S_i;\{S_j\}_{j\in\mathcal{V}_c}\cup\{S_j\}_{j\in \mathcal{V}\backslash \mathcal{V}_p},\\
    &\{Y_j\}_{j\in \mathcal{V}}\mid Y_i=1)\label{eq.theo13}
\end{align}
  \label{theo:I}
    \end{theorem}

\begin{proof}
    See Appendix~\ref{app:max_prime}.
\end{proof}

Theorem~\ref{thm.max_prime} separates two regimes. If $Y_i=0$, then \eqref{eq.condition_P} fails at some iteration and the max operator in \eqref{eq.xi_up} enforces $x_i^{(t)}=s_i$, which yields potentially large leakage. If $Y_i=1$, on the other hand,  the max operator never selects $s_i$ during the run, and the leakage is limited to what is implied by the remaining observations.

\subsubsection{Optimality analysis}\label{ss.max_sys}
We now connect the above characterization to the ideal-world benchmark in Section~\ref{sec.ideal}. Under global eavesdropping, revealing the identity of the maximizer is unavoidable, but additional disclosure (e.g., revealing non-maximizer values via $x_i^{(t)}=s_i$) is not. We therefore aim to enforce \eqref{eq.condition_P} for all non-maximizer nodes so that only the ideal-world leakage remains.

To ensure $Y_i=1$, we require the consensus term to start and remain larger than $s_i$. This relies on the initialization
$$\frac{-1 - \sum_{j \in \mathcal{N}_i} a_{ij} z_{i \mid j}^{(0)}}{c d_i} > s_i.$$

By choosing $\mu_z \gg s_{\max}$ and $\mu_z \gg \sigma_z$, and initializing $z_{i|j}^{(0)} \approx \pm \mu_z$ based on $\text{sign}(a_{ij})$, we ensure the numerator is large and positive. Consequently, in early iterations, the max operator is inactive ($x_i^{(t)} > s_i$), and the update simplifies to
$$x_i^{(t)} = \frac{-1 - \sum_{j \in \mathcal{N}_i} a_{ij} z_{i \mid j}^{(t)}}{c d_i }.$$

This linear update dynamic corresponds to standard PDMM for average consensus. Because the objective is unbounded below, $x_i^{(t)}$ decreases
at a rate governed by $c$ until it reaches the vicinity of the global maximum. For any node $i$ that does not hold the maximum, the limit point is strictly greater than $s_i$, so \eqref{eq.condition_P} can hold for all $t$ and $s_i$ is never directly revealed. In contrast, for a node attaining $s_{\max}$, the iterate necessarily hits $s_i$ to produce the correct output, which is consistent with the unavoidable leakage identified in Section~\ref{sec.ideal}.

\begin{lemma}[Achievability of the Ideal-World Privacy Bound]\label{lm.1}
The proposed algorithm can attach the ideal-world privacy bound 
\begin{align*}
    I(S_i;\mathcal{O}) = I\!\left(S_i;\ S_{\max},\ \left\{S_j\right\}_{j\in\mathcal{V}_c},\ \left\{\mathbb{I}_{\{S_j = S_{\max}\}}(S_j)\right\}_{j\in \mathcal{V}}\right),
\end{align*}
provided that:
\begin{enumerate}
    \item $c \gg 1$ (a sufficiently small step-size),
    \item all nodes $j$ with $s_j \neq \max_{k \in \mathcal{V}} s_k$ belong to $\mathcal{V}_p$, i.e., $P(Y_j = 1) = 1$ for every such node;
\end{enumerate}
\end{lemma}
\begin{proof}
See Appendix~\ref{app:lm}.
\end{proof}

\begin{remark}
Condition~(1) guarantees that the update steps of $x_i^{(t)}$ are sufficiently small so that trajectories approach $s_{\max}$ without large oscillation. In practice, increasing $c$ must be accompanied by a compatible choice of the initialization (sufficiently large $\mu_z$ compared to $c$) to satisfy Condition~(2).
\end{remark}

\subsection{Minimum consensus}
Minimum consensus is the counterpart of maximum consensus. Since the privacy mechanism and analysis follow the same template, we only highlight the difference in the local constraint. The minimum consensus problem can be formulated as
\begin{equation*}
\begin{array}{ll} \underset{\{x_i:{i \in \mathcal{V}}\}}{\text{minimize}}  & {\displaystyle \sum_{i\in {\mathcal V}} (-x_i+{\cal I}_{\{x_i\leq s_i\}}(x_i)),} \\\rule[4mm]{0mm}{0mm}
\text{subject to} & x_i -x_j = 0, \quad (i,j)\in \mathcal  E.
\end{array}
\end{equation*}
To implement minimum consensus, we only need to modify the update rule in \eqref{eq.xi_up} as 
\begin{equation*}x_{  i}^{(t)} =\min\left(\frac{ 1- \sum_{  j \in {\mathcal N}_i}  a_{ij}z_{  i|j}^{(t)}}{cd_i},s_i\right).\end{equation*}
The corresponding privacy intuition is symmetric: one aims to control the trajectory so that non-minimizer nodes do not have iterates that directly equal their private values.

\subsection{Top-$K$ aggregation}
Top-$K$ aggregation generalizes maximum consensus: the goal is to identify the $K$ largest values in the network. A natural approach is iterative: first identify the maximum, then modify the constraint (indicator term) to exclude already-identified maxima, and repeat.

Define the maximum value as $s_{\max} = s_{(1)} = \max_{i\in\mathcal{V}} s_i$, and let $s_{(k)}$ denote the $k$-th sorted value in the set $\{s_i : i \in \mathcal{V}\}$.
In top-$K$ aggregation, after identifying the maximum value, the indicator function for the subsequent optimization round must be adjusted. For instance, when searching for the second largest value, the indicator function is redefined as $\mathcal{I}_{\{x_i < s_i\} \cup \{s_i < s_{\max}\}}(x_i)$.

Accordingly, the algorithm proceeds in $K$ major rounds to sequentially identify the $K$ largest values. Let $\mathcal{S}$ denote the pool of values identified in the previous rounds (i.e., the current top-$k$ values where $k < K$). During the $x$-update step of the $k$-th round, nodes holding the private value already in $\mathcal{S}$ bypass the maximum operator to allow their states to evolve freely. However, due to the limited number of iterations $t_{\max}$ in each round, the consensus values may not perfectly match  $s_{(k)}$. We, therefore, introduce a tolerance $\epsilon > 0$ and define the set of identified values 
as $\mathcal{S}_{\epsilon}=\{\, s + \delta \mid s \in \mathcal{S},\ \delta \in [-\epsilon, \epsilon] \,\}$. The update rule in \eqref{eq.xi_up} is then modified as:
\begin{align}\label{eq.topk_xi}
     x_{i}^{(t)} = \begin{cases} 
     \frac{-1 - \sum_{j \in \mathcal{N}_i} a_{ij} z_{i|j}^{(t)}}{c d_i}, & \text{if } s_i\in\mathcal{S}_{\epsilon} \\
     \max\left(\frac{-1 - \sum_{j \in \mathcal{N}_i} a_{ij} z_{i|j}^{(t)}}{c d_i}, s_i\right), & \text{otherwise}
    \end{cases}
\end{align}
where $s_i\in\mathcal{S}_{\epsilon}$ indicates that node $i$ holds one of the values already identified in the pool $\mathcal{S}$ within a small tolerance. After $t_{\max}$ iterations of \eqref{eq.topk_xi}, the next maximum value is identified, added to $\mathcal{S}$, and the process repeats for $K$ rounds. The full procedure is summarized in Algorithm~\ref{alg:topk}.

From the ideal-world viewpoint, top-$K$ inevitably reveals membership/identity information about the top-$K$ set. The best-case privacy behavior is therefore analogous to maximum consensus: additional leakage beyond the top-$K$ values and their associated membership indicators can be avoided when the iterates of non-top-$K$ nodes do not directly hit their private values.

\begin{algorithm}[ht]
  \caption{Privacy preserving top-$K$ aggregation}
  \label{alg:topk}
  \begin{algorithmic}
     \ForAll{ $i \in \mathcal{V}, j \in \mathcal{N}_i$,}
     \State Randomly initialize $z_{i\mid j}^{(0)}\sim\mathcal{N}(\pm\mu_z,\sigma_z^2)$ \Comment{Initialization}
     \State $\text{Node}_j \leftarrow \text{Node}_i(z_{  i|j}^{(0)})$
     \State Initialize $\mathcal{S}$ as an empty set
     \EndFor
     \State Run the maximum consensus algorithm for $t_{\max}$ iterations to obtain an approximate solution $x^{(t_{\max})}=\hat{s}_{\max}$
     \State Add $\hat{s}_{\max}$ to set $\mathcal{S}$
     \For{$k=1,2,\cdots,K$} 
     \State  Initialize $z_{i\mid j}^{(0)}:=z_{i\mid j}^{(t_{\max})}$ using the $z$ values from the previous round
          \For{$t=0,1,\cdots, t_{\max}-1$} 
            \ForAll{$i \in \mathcal{V}$ } 
\State Update $x_{  i}^{(t)}$ with \eqref{eq.topk_xi}
\State Update $\forall j\in \mathcal{N}_i: z_{  j|i}^{(t+1)}$ with \eqref{eq.zji_up}
                \State $\text{Node}_{j\in \mathcal{N}_i} \leftarrow \text{Node}_i(x_i^{(t)})$\Comment{Broadcast}
                \ForAll{ $j \in \mathcal{N}_i$} 
        \State Update $z_{  j|i}^{(t+1)}$ with \eqref{eq.zji_up}
            \EndFor
      \EndFor
      \EndFor
      \State Add $x^{(t_{\max})}=\hat{s}_{(k)}$ to set $\mathcal{S}$
      \EndFor
  \end{algorithmic}
\end{algorithm}

\section{Robust statistical aggregation}\label{sec.robust}
This section studies robust statistical aggregators, focusing on the coordinate-wise median/quantiles and trimmed mean. These methods are widely used to mitigate outliers and adversarial perturbation~\cite{wang2025federated,cambus2025coordinate}. Compared to extrema, their privacy behavior is more nuanced: the ideal-world discussion in Section~\ref{sec.ideal} indicates that under eavesdropping, median-type aggregation unavoidably reveals not only the identity of the median holder but also side information about whether each node lies above or below the median. The purpose of this section is to quantify this leakage for our protocol and to clarify under what conditions the algorithm can approach the ideal-world benchmark.

\subsection{Median consensus\label{sec.median}}
We consider distributed median consensus. Median aggregation is attractive because it is robust to outliers: unlike averaging, it is not dominated by a small number of extreme values. This robustness, however, comes with different privacy tradeoffs than maximum/minimum aggregation because nodes may approach the consensus value from both sides.
\subsubsection{Proposed approach}
The $x$-update can be written as the piecewise form in \eqref{eq.x_update_eq}; the full procedure is summarized in Algorithm~\ref{alg:med_primal}.
\begin{align}\label{eq.x_update_eq}
     x_{  i}^{(t)}=\left\{ \begin{array}{ll}
         \frac{ -1- \sum_{  j \in {\mathcal N}_i} a_{  ij} z_{  i|j}^{(t)}}{cd_i},  &\text{if   } \frac{ -1- \sum_{  j \in {\mathcal N}_i} a_{  ij} z_{  i|j}^{(t)}}{cd_i}> s_i,\\
         \frac{ 1- \sum_{  j \in {\mathcal N}_i} a_{  ij} z_{  i|j}^{(t)}}{cd_i}, &\text{if   } \frac{ 1- \sum_{  j \in {\mathcal N}_i} a_{  ij} z_{  i|j}^{(t)}}{cd_i}< s_i,\\
         s_i, &\text{otherwise.} \rule[4mm]{0mm}{0mm}\\
    \end{array}\right.
\end{align}

\begin{algorithm}[ht]
  \caption{Privacy-preserving median consensus}
  \label{alg:med_primal}
  \begin{algorithmic}
     \ForAll{ $i \in \mathcal{V}, j \in \mathcal{N}_i$,}
     \State Randomly initialize $z_{i\mid j}^{(0)}$ \Comment{Initialization}
     \State $\text{Node}_j \leftarrow \text{Node}_i(z_{  i|j}^{(0)})$
     \EndFor
          \For{$t=0,1,\cdots,t_{\max}-1$} 
            \ForAll{$i \in \mathcal{V}$ } 
\State Update $x_{  i}^{(t)}$ with \eqref{eq.x_update_eq}
\State Update $\forall j\in \mathcal{N}_i: z_{  j|i}^{(t+1)}$ with \eqref{eq.zji_up}
                \State $\text{Node}_{j\in \mathcal{N}_i} \leftarrow \text{Node}_i(x_i^{(t)})$\Comment{Broadcast}
                \ForAll{ $j \in \mathcal{N}_i$} 
        \State Update $ z_{  j|i}^{(t+1)}$ with \eqref{eq.zji_up}
            \EndFor
      \EndFor
      \EndFor
  \end{algorithmic}
\end{algorithm}

\subsubsection{Performance Analysis}
We use the same observation model $\mathcal{O}$ as in \eqref{eq.adv_info_max} (passive corruption plus global eavesdropping under plaintext communication). In contrast to maximum consensus, the piecewise structure in \eqref{eq.x_update_eq} introduces \emph{two} inequality regimes, corresponding to approaching the median from above or from below. Building on the refined argument in~\cite{yu2025optimal}, the following theorem decomposes the leakage accordingly.

\begin{theorem}\label{thm.med_l1}
Let ${\mathcal V}_{p1}$ denote the set of nodes $i\in \mathcal V _{p1}$ that satisfy
    \begin{equation*}
    \frac{ -1- \sum_{  j \in {\mathcal N}_i}  a_{  ij}z_{  i|j}^{(t)}}{cd_i}>s_i
  \end{equation*}
  for all $t\in \mathcal T$.
Similarly, let ${\mathcal V}_{p2}$ denote the set of nodes satisfying
    \begin{equation*}
    \frac{ 1- \sum_{  j \in {\mathcal N}_i}  a_{  ij}z_{  i|j}^{(t)}}{cd_i}<s_i.
  \end{equation*} 

 Let $Y_i = \mathbb{I}_{\{i\in\mathcal{V}_{p1}\}}(S_i)$ and $Y^\prime_i = \mathbb{I}_{\{i\in\mathcal{V}_{p2}\}}(S_i)$, then we have 
\begin{align*}
    I(S_i;\mathcal{O})=&I(S_i;Y_i,Y_i^\prime)\\
    &+P(Y_i=0,Y_i^\prime=0)I(S_i;\mathcal{O}|Y_i=0,Y_i^\prime=0)\\
    &+P(Y_i=1,Y_i^\prime=0)I(S_i;\mathcal{O}|Y_i=1,Y_i^\prime=0)\\
    &+P(Y_i=0,Y_i^\prime=1)I(S_i;\mathcal{O}|Y_i=0,Y_i^\prime=1)
\end{align*}
where
\begin{align*}
    I(S_i;\mathcal{O}|Y_i=0,Y_i^\prime=0)\leq&I(S_i;S_i\mid Y_i=0,Y_i^\prime=0)\\
    I(S_i;\mathcal{O}|Y_i=1,Y_i^\prime=0)\leq&I(S_i;\{S_j\}_{j\in\mathcal{V}_c}\cup\{S_j\}_{j\in \mathcal{V}\backslash \{\mathcal{V}_{p1}\cup\mathcal{V}_{p2}\}},\\
    &\{Y_j,Y_j^\prime\}_{j\in \mathcal{V}}\mid Y_i=1,Y_i^\prime=0)\\
    I(S_i;\mathcal{O}|Y_i=0,Y_i^\prime=1)\leq&I(S_i;\{S_j\}_{j\in\mathcal{V}_c}\cup\{S_j\}_{j\in \mathcal{V}\backslash \{\mathcal{V}_{p1}\cup\mathcal{V}_{p2}\}},\\
    &\{Y_j,Y_j^\prime\}_{j\in \mathcal{V}}\mid Y_i=0,Y_i^\prime=1)
\end{align*}
  \end{theorem}

Since $\mathcal{V}_{p1}$ and $\mathcal{V}_{p2}$ are disjoint, $P(Y_i=1, Y_i^\prime=1)=0$, which yields a decomposition analogous to Theorem~\ref{theo:I}. The case $(Y_i,Y_i^\prime)=(0,0)$ corresponds to trajectories that enter the ``median window'' in \eqref{eq.x_update_eq}; in that case, it may happen that $x_i^{(t)}=s_i$ for some $t$, which leads to complete disclosure of $s_i$. Otherwise, $s_i$ acts only through the switching behavior of \eqref{eq.x_update_eq} and the observation constrains $s_i$ to a range determined by the crossing event. In our experiments (see Section~\ref{sec.experiment}), the direct-hit behavior is the dominant mode, and the inequalities above are therefore often tight in practice.

The theorem also provides an interpretation consistent with the ideal-world view in Section~\ref{sec.ideal}: nodes that avoid direct disclosure fall into two trajectory classes. High-value nodes in $\mathcal{V}_{p1}$ approach the consensus value from above, while low-value nodes in $\mathcal{V}_{p2}$ approach from below. This mirrors the unavoidable ``above/below the median'' side information present in the ideal world under eavesdropping.

\subsubsection{Optimality analysis}\label{ss.med_sys}
Compared to maximum consensus, attaining the ideal-world bound for median consensus is inherently more challenging, since private values lie on both sides of the consensus value and the iteration may approach the limit from either direction. The key to protecting a private value $s_i$ is to ensure that it never enters the interval 
$[\frac{ -1- \sum_{  j \in {\mathcal N}_i} a_{  ij} z_{  i|j}^{(t)}}{cd_i},\frac{ 1- \sum_{  j \in {\mathcal N}_i} a_{  ij} z_{  i|j}^{(t)}}{cd_i}]$, whose width is $\tfrac{2}{c d_i}$; see Fig.~\ref{fig:medwindow} for an illustrative example.
Doing so typically requires some prior knowledge about whether a node lies below or above the median, so that one can initialize the trajectory to approach the limit from the correct side and avoid crossing the sensitive window.

For example, if it is known that $s_i$ is below the median, one may initialize $x_i^{(0)}$ sufficiently above the median so the trajectory approaches from above; similarly for values above the median. Under such a setup, when $c \gg 1$ and the convergence window shrinks, the indicators $Y_i$ and $Y_i^\prime$ become closely aligned with the ideal-world side information
$\mathbb{I}_{\{S_i \le S_{\med}\}}(S_i)$ and
$\mathbb{I}_{\{S_i \ge S_{\med}\}})(S_i)$.
In this idealized regime, the protocol can behave close to the ideal world, protecting all nodes except the one attaining the median. Without such prior information, however, one cannot in general steer all trajectories to avoid the median window simultaneously, and the achievable leakage typically exceeds the ideal-world benchmark.

\begin{figure}
    \centering
    \includegraphics[width=0.99\linewidth]{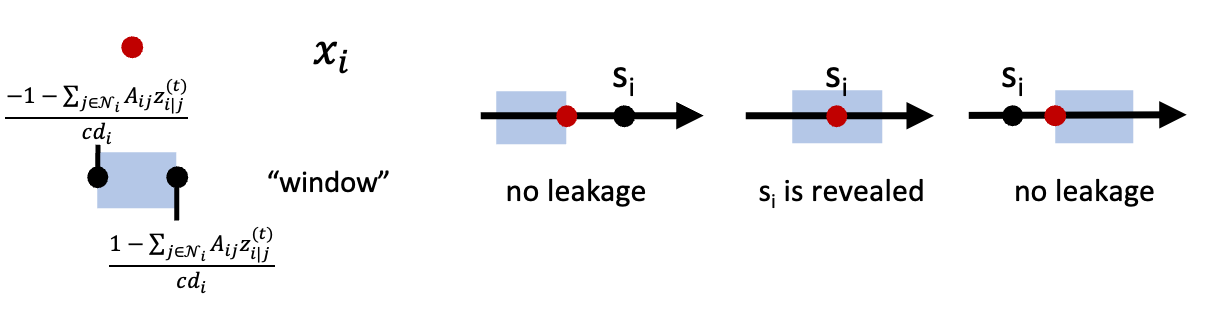}
    \caption{A schematic illustration of the leakage of $s_i$ during the $x$-update of the median consensus protocol.}
    \label{fig:medwindow}
\end{figure}

\subsection{Quantile consensus}
Quantile consensus is a natural generalization of median consensus. In a network of nodes, each node holds a value, and the goal is for the network to reach agreement on a specific quantile of all the node values. This allows for flexible aggregation beyond the minimum, maximum, or median, making it useful in applications where a specific percentile or threshold is of interest.

We define the quantile loss as
\begin{align*}
\rho_q(r) =
\begin{cases} 
q \cdot r, & \text{if } r \ge 0, \\ 
(q-1) \cdot r, & \text{if } r < 0,
\end{cases}
\end{align*}
where $q \in (0,1)$ denotes the desired quantile.

Then, the quantile consensus problem can be formulated as
\begin{equation*}
\begin{array}{ll} 
\underset{\{x_i:{i \in \mathcal{V}}\}}{\text{minimize}} & {\displaystyle \sum_{i\in \mathcal{V}} \rho_q(x_i - s_i),} \\\rule[4mm]{0mm}{0mm}
\text{subject to} & x_i - x_j = 0, \quad (i,j) \in \mathcal{E}.
\end{array}
\end{equation*}

To implement quantile consensus, we modify the update rule in \eqref{eq.x_update_eq} as
\begin{align*}
x_i^{(t)} =
\begin{cases}
\frac{-q - \sum_{j \in \mathcal{N}_i} A_{ij} z_{i|j}^{(t)}}{c d_i}, & \text{if } \frac{-q - \sum_{j \in \mathcal{N}_i} A_{ij} z_{i|j}^{(t)}}{c d_i} > s_i, \\
\frac{1-q - \sum_{j \in \mathcal{N}_i} A_{ij} z_{i|j}^{(t)}}{c d_i}, & \text{if } \frac{1-q - \sum_{j \in \mathcal{N}_i} A_{ij} z_{i|j}^{(t)}}{c d_i} < s_i, \\
s_i, & \text{otherwise.} \rule[4mm]{0mm}{0mm}
\end{cases}
\end{align*}

Similar to median consensus, the key to preserving privacy in quantile consensus is to prevent the private value $s_i$ from falling within the interval
$\left[ \frac{-q - \sum_{j \in \mathcal{N}_i} A_{ij} z_{i|j}^{(t)}}{c d_i}, \frac{1-q - \sum_{j \in \mathcal{N}_i} A_{ij} z_{i|j}^{(t)}}{c d_i} \right].$
By controlling the convergence trajectory relative to this interval, sensitive values can be protected while allowing the network to reach agreement on the desired quantile.

\subsection{Trimmed mean aggregation}

Trimmed mean aggregation is a robust statistical method used to compute an average while mitigating the influence of outliers in many robust federated learning algorithms~\cite{wang2025federated}, thereby providing resilience against Byzantine attacks. It works by dropping the largest and smallest $K$ values (where $K$ is determined by the trim ratio), and then averaging the remaining ones. Formally, the trimmed mean estimate can be expressed as
\[
x^*=\frac{n\overline{s}-\sum_{k=1}^K(s_{(k)}+ s_{(n+1-k)})}{n-2K},
\]
where $\overline{s}$ is the average value of all nodes.
Therefore, we can approximate the trimmed mean by combining the top-$K$ extrema with the global average. Once each node reports its top-$K$ largest and smallest values together with its mean, the trimmed mean estimate can be directly computed using the formula above.

In prior work~\cite{Jane2020TIFS}, it has been shown that the theoretical information leaked by average consensus in the ideal world is given by
\begin{align*}
    I(S_i;\{S_j\}_{j\in\mathcal{V}_c},\{\sum_{j\in\mathcal{V}_{h,m}}S_j\}_{m=1,2,\cdots,M}),
\end{align*}
where $\mathcal{V}_{h,m}$ denotes the $m$-th connected component formed by the honest nodes after removing all corrupted nodes from the graph. Moreover, this information-theoretic bound can be achieved using the privacy-preserving A/PDMM algorithm~\cite{Jane2020TIFS} when $\sigma_z\rightarrow\infty$.

By incorporating the information additionally leaked through the top-$K$ values, we obtain that the achievable and theoretically minimal information leakage of trimmed mean aggregation is
\begin{align*}
    &I(S_i;\{S_j\}_{j\in\mathcal{V}_c},\{\mathbb{I}_{S_j= S_{(n-K+1)}}(S_j),\mathbb{I}_{S_j= S_{(K)}}(S_j)\}_{j\in \mathcal{V}}\\
    &\{\sum_{j\in\mathcal{V}_{h,m}}S_j\}_{m=\{1,2,\cdots,M\}},\{S_{(k)},S_{(n+1-k)}\}_{k=\{1,2,\cdots,K\}}).
\end{align*}

\section{Experimental validations}\label{sec.experiment}
Due to space constraints, we present only the key simulation results for maximum and median consensus in this section. Results for minimum consensus, top-$K$ aggregation, quantile consensus, and trimmed mean consensus, together with additional figures, are provided in the supplementary material\footnote{The code is publicly available at \url{https://github.com/Wenrui-Yu/privacy-preserving-nonlinear-aggregation}.}.
\subsection{Experimental setup}\label{ss:exp_setup}
Unless otherwise specified, we use random geometric graphs (RGGs)~\cite{penrose2003random} with $n=15$ nodes uniformly distributed over a $1\times 1$ square as the test topology. Each node $i$ holds a private value $s_i$ independently sampled from a Gaussian distribution $\mathcal{N}(0,1)$. To consistently identify which private values are more prone to leakage across Monte Carlo trials, we relabel agents according to the order statistics of their private values in each trial. Specifically, the sampled values $\{s_i\}_{i=0}^{n-1}$ are sorted in descending order and assigned to node indices, so that node $0$ holds the largest value, node $1$ the second largest, and so on. This relabeling is equivalent to a permutation of node indices, does not alter the joint distribution or independence of the private values, and is performed independently of the communication topology. Its sole purpose is to simplify the visualization of leakage statistics across trials.

For comparison, we consider three approaches:
\begin{itemize}
    \item \textbf{Non-private counterpart:} the original consensus dynamics without any masking or perturbation at initialization (i.e., all $z_{i\mid j}^{(0)}=0$ or equivalently, $\mu_z=0$ and $\sigma_z=0$ in our notation).
    \item \textbf{Existing DP-based approach:} following the noisy initialization strategy commonly adopted in the literature\footnote{To ensure a fair comparison, we employ Gaussian noise as the perturbation, whereas the original work utilizes Laplacian noise. Both types of noise exhibit the same trend, as demonstrated in Fig.~8 and 9 in the Supplementary.}~\cite{wang2018differentially}, we perturb private values once at $t=0$ via
    \begin{equation}\label{eq:dp_init}
        \hat{s}_i = s_i + \mathcal{N}(0,\sigma_s^2),
    \end{equation}
    where $\sigma_s$ denotes the standard deviation of the injected noise. The subsequent dynamics are identical to those of the non-private counterpart.
    \item \textbf{Proposed approach:} the privacy-preserving initialization and dynamics described in Sections~\ref{ss.max_sys} and~\ref{ss.med_sys}, parameterized by $(c, \mu_z)$.
\end{itemize}

We evaluate these approaches via (i) \emph{output accuracy} using MSE metric, and (ii) \emph{individual privacy} using mutual information, as defined in Section~\ref{ss.eval}. All results are averaged over $1000$ Monte Carlo trials, and mutual information is estimated using the NPEET toolbox~\cite{npeet}.

\subsection{Output accuracy: exact aggregation without utility loss}\label{ss:accuracy}
Fig.~\ref{fig:mse} reports the output accuracy (MSE) of maximum and median consensus under the proposed approach, the non-private counterpart, and the existing DP-based approach. The proposed approach is \emph{exact}: it converges to the true aggregate with vanishing MSE, matching the non-private counterpart. In contrast, the existing DP-based approach exhibits a nonzero steady-state MSE due to the injected noise, the error increases with the noise magnitude. This highlights a fundamental privacy-accuracy trade-off in existing DP-based approaches, which is avoided by the proposed approaches through exact aggregation.

\begin{figure}[ht]
    \centering
    \begin{minipage}{0.24\textwidth}
        \centering
        \includegraphics[width=\linewidth]{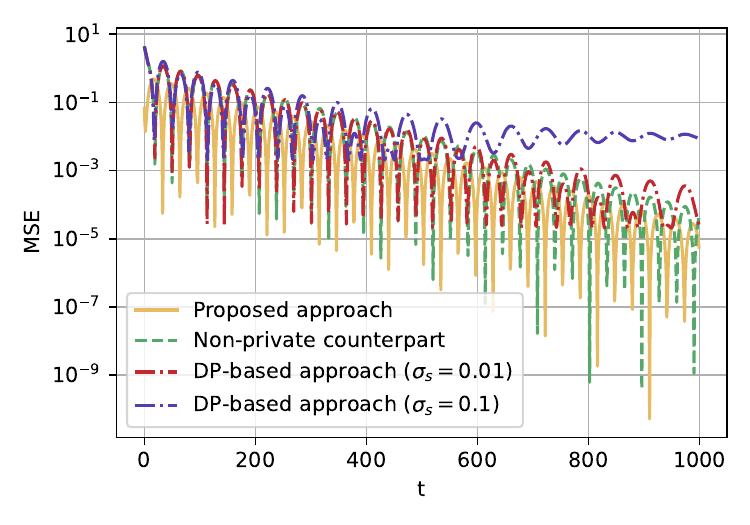}
        \\[-1ex] 
        {\small (a) Maximum consensus}
    \end{minipage}
    \hfill
    \begin{minipage}{0.24\textwidth}
        \centering
        \includegraphics[width=\linewidth]{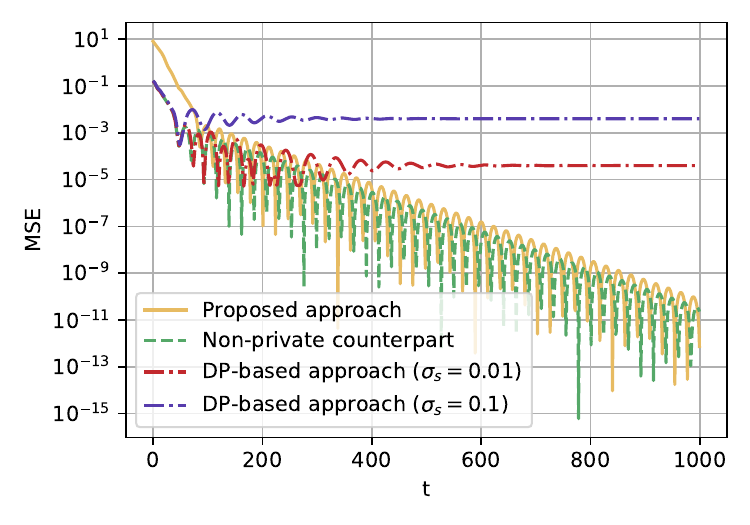}
        \\[-1ex]
        {\small (b) Median consensus}
    \end{minipage}
    \caption{Output accuracy (MSE) of maximum and median consensus under the proposed approach, the non-private counterpart, and the existing DP-based approaches.}
    \label{fig:mse}
\end{figure}

\subsection{Individual privacy}\label{ss:privacy}
We now evaluate individual privacy, focusing on how much information about each private value $s_i$ can be inferred from the observed state evolution. Although the privacy bounds derived in Theorem \ref{thm.max_prime} and \ref{thm.med_l1} involve multiple parameters and case distinctions, their operational meaning admits a simple interpretation: individual information leakage occurs when a \emph{matching event} happens, i.e., when $x_i^{(t)} = s_i$ for some time~$t$. Once such an event occurs, the corresponding private value is directly exposed through the state trajectory. This interpretation motivates our per-node privacy metrics. We focus on the state trajectory $\{x_i^{(t)}\}_{t\ge 0}$ of each node and quantify leakage using the normalized mutual information (NMI)
$\frac{I(S_i;X_i^{(t)})}{I(S_i;S_i)}.$ In addition, we report the \emph{leakage probability}, defined as the leakage count normalized by the total number of Monte Carlo trials ($1000$), where a leakage corresponds to the occurrence of a matching event, i.e., $x_i^{(t)} = s_i$ for some $t$. 

Specifically, in what follows we will first examine whether the empirical leakage behavior is consistent with the privacy bounds predicted by the theoretical analysis. We then compare the privacy of the proposed approach with the non-private counterpart and the existing DP-based approach. 

\subsection*{Part I. Consistency with theoretical privacy bounds}\label{ss:privacy_bound}
\textbf{Maximum consensus.} Fig.~\ref{fig:nmi_maximum} confirms the results in Theorem~\ref{thm.max_prime} and Lemma~\ref{lm.1}.  We highlight three observations.
First, the node holding the maximum value inevitably incurs full leakage at convergence (NMI reaches $1$), since the aggregation output itself is disclosed by design.
Second, nodes whose private values are closer to the maximum are more prone to leakage (see Fig.~\ref{fig:nmi_maximum} (a) and (b)). This is expected by the transient envelope characterized in~\eqref{eq.xi_up}: during oscillations around the maximum, trajectories are more likely to intersect values near the target, increasing the chance of a matching event $x_i^{(t)}=s_i$. 
Third, avoidable leakage can be controlled by parameters. By increasing $c$ and choosing a compatible initialization offset $\mu_z$, the sufficient conditions in Lemma~\ref{lm.1} are effectively satisfied, suppressing matching events.  Consistent with this prediction, Fig.~\ref{fig:nmi_maximum}(a)-(d) show that as $c$ is increased progressively (together with a compatible choice of $\mu_z$), the leakage probability decreases monotonically and eventually becomes negligible. 

\begin{figure}[ht]
    \centering
    \begin{minipage}{0.24\textwidth}
        \centering
        \includegraphics[width=\linewidth]{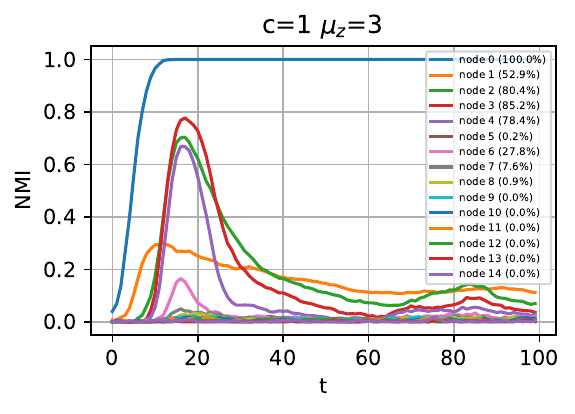}
        \\[-1ex]{\small (a)}
    \end{minipage}
    \begin{minipage}{0.24\textwidth}
        \centering
        \includegraphics[width=\linewidth]{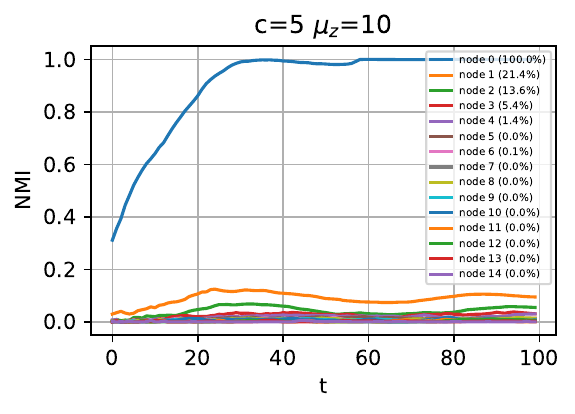}
        \\[-1ex]{\small (b)}
    \end{minipage}
       \hfill
    \begin{minipage}{0.24\textwidth}
        \centering
        \includegraphics[width=\linewidth]{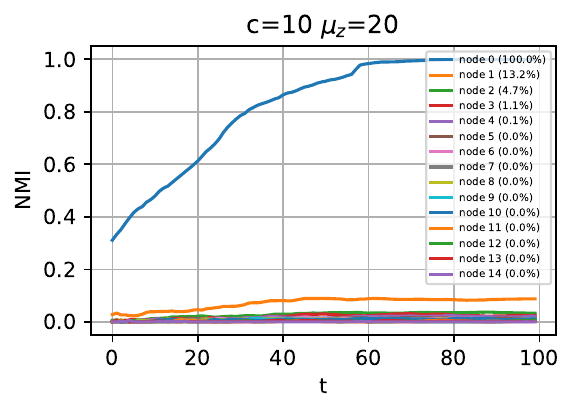}
        \\[-1ex]{\small (c)}
    \end{minipage}
    \hfill
    \begin{minipage}{0.24\textwidth}
        \centering
        \includegraphics[width=\linewidth]{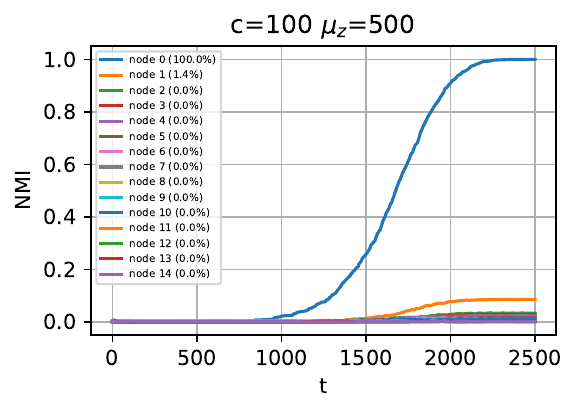}
        \\[-1ex]{\small (d)}
    \end{minipage}
    \caption{Maximum consensus: NMI $\frac{I(S_i;X_i^{(t)})}{I(S_i;S_i)}$ under different initialization settings. Values in parentheses indicate the empirical leakage probability of each node over $1000$ Monte Carlo trials.}
    \label{fig:nmi_maximum}
\end{figure}

\textbf{Median consensus.}
Fig.~\ref{fig:nmi_median_prior} is to validate Theorem~\ref{thm.med_l1}. We also highlight three observations.
First, the node holding the median value inevitably incurs full leakage at convergence (NMI reaches $1$), as expected. 
Second, similar to maximum consensus, non-median nodes whose private values are closer to the median exhibit higher leakage probabilities (see Fig.~\ref{fig:nmi_median_prior} (a),(b)), because values near the target are more likely to be intersected during the transient evolution, thereby triggering a matching event $x_i^{(t)}=s_i$.

More importantly, we note that Theorem~\ref{thm.med_l1} highlights a fundamental difficulty specific to median consensus: to achieve the ideal-world privacy bound one needs to satisfy \emph{two} privacy conditions simultaneously (encoded by the sets $\mathcal{V}_{p1}$ and $\mathcal{V}_{p2}$), reflecting constraints from both sides of the median. Without prior information to guide the choice of initialization, it is very challenging to keep all trajectories away from the sensitive interval
$
\Big[\tfrac{ -1- \sum_{  j \in {\mathcal N}_i} a_{  ij} z_{  i|j}^{(t)}}{c d_i},\;
      \tfrac{ 1- \sum_{  j \in {\mathcal N}_i} a_{  ij} z_{  i|j}^{(t)}}{c d_i}\Big],
$
thus matching events (and hence leakage) may still occur for some non-median nodes, e.g., Fig.~\ref{fig:nmi_median_prior} (a),(b).

We therefore consider two representative cases. In Fig.~\ref{fig:nmi_median_prior}(c), we deliberately enforce only a one-sided condition in Theorem~\ref{thm.med_l1} through the initialization choice, which guarantees privacy for the nodes on one side of the median and yields a partial (roughly half) privacy guarantee. In Fig.~\ref{fig:nmi_median_prior} (d), we instead assume prior  information and construct a bimodal setting (used only in this subplot) with $\lfloor n/2 \rfloor$ nodes drawn from $\mathcal{N}(-5,1)$, $\lfloor n/2 \rfloor$ nodes drawn from $\mathcal{N}(5,1)$, and one node drawn from $\mathcal{N}(0,1)$. This prior information enables a consistent initialization near the median, so that trajectories avoid the sensitive interval for all non-median nodes, leading to zero avoidable leakage in Fig.~\ref{fig:nmi_median_prior}(d).

\textbf{The impact of topology.}
Network topology also plays a significant role in privacy protection. As shown in Fig.~\ref{fig:nmi_mediantopo}(a) and (b), for maximum consensus, denser networks—with larger node degrees $d_i$—exhibit lower leakage frequencies for non-maximum nodes. As indicated by~\eqref{eq.xi_up}, increasing $d_i$ leads to a smaller initialization value $x^{(0)}$ and a slower convergence rate, both of which reduce the likelihood of matching events. A similar trend is observed for median consensus. Figs.~\ref{fig:nmi_mediantopo}(c) and (d) show that denser networks also yield lower leakage frequencies. This behavior is consistent with the theoretical analysis in Section~\ref{ss.med_sys}: the width of the sensitive leakage window, given by $\frac{2}{c d_i}$, shrinks as the node degree $d_i$ increases, thereby reducing the probability that state trajectories enter the critical region and trigger information leakage.

Hence, the above results provide strong empirical validation of the theoretical analysis, confirming that the privacy bounds and leakage mechanisms characterized in Theorems~\ref{thm.max_prime} and~\ref{thm.med_l1} accurately predict the observed privacy behavior under varying parameters and network topologies.

\begin{figure}[ht]
    \centering
    \begin{minipage}{0.24\textwidth}
        \centering
        \includegraphics[width=\linewidth]{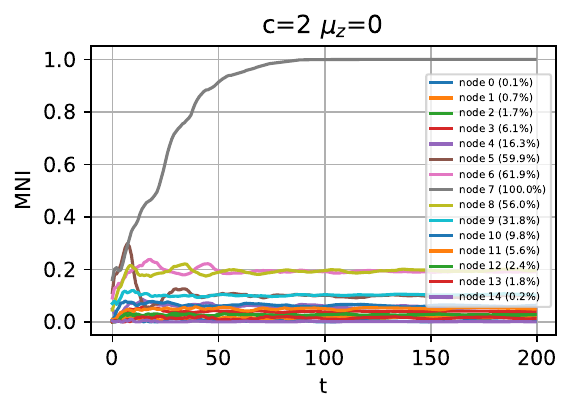}
        \\[-1ex] 
        {\small (a)}
    \end{minipage}
    \hfill
    \begin{minipage}{0.24\textwidth}
        \centering
        \includegraphics[width=\linewidth]{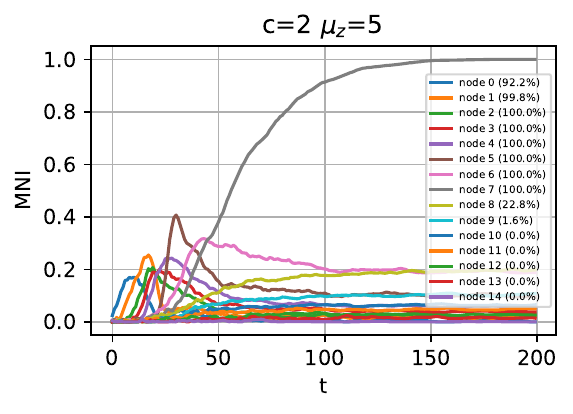}
        \\[-1ex]{\small (b)}
    \end{minipage}
    \begin{minipage}{0.24\textwidth}
        \centering
        \includegraphics[width=\linewidth]{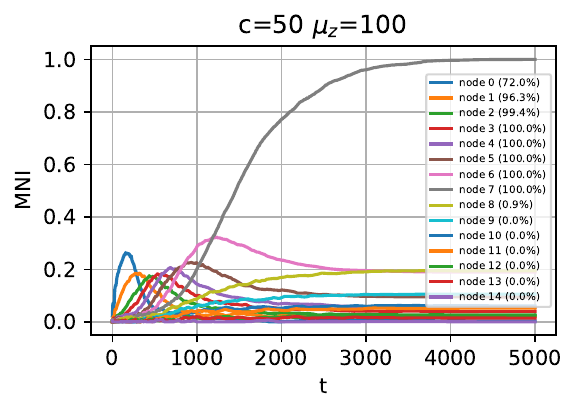}
        \\[-1ex]{\small (c)}
    \end{minipage}
    \hfill
    \begin{minipage}{0.24\textwidth}
        \centering
        \includegraphics[width=\linewidth]{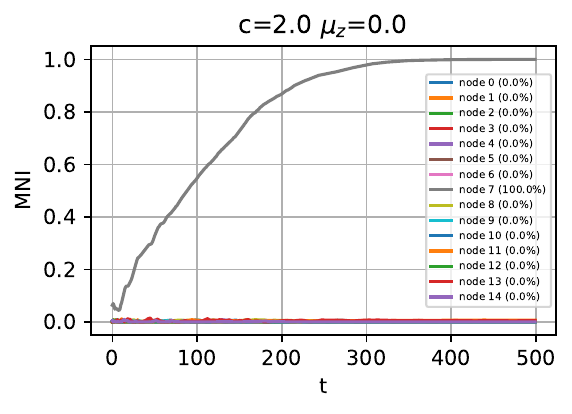}
        \\[-1ex]{\small (d)}
    \end{minipage}
    \caption{Medain consensus: NMI $\frac{I(S_i;X_i^{(t)})}{I(S_i;S_i)}$ under different initializations. Values in parentheses indicate the empirical leakage probability of each node over $1000$ Monte Carlo trials.}
    \label{fig:nmi_median_prior}
\end{figure}

\begin{figure}[ht]
    \centering
    \begin{minipage}{0.24\textwidth}
        \centering
        \includegraphics[width=\linewidth]{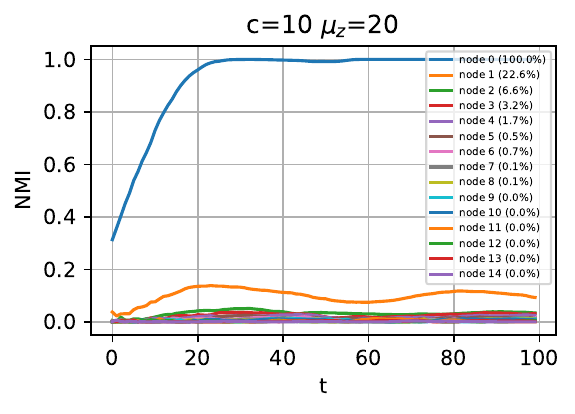}
        \\[-1ex]{\footnotesize (a) Maximum, ring ($d_i=2$)}
    \end{minipage}
    \hfill
    \begin{minipage}{0.24\textwidth}
        \centering
                \includegraphics[width=\linewidth]{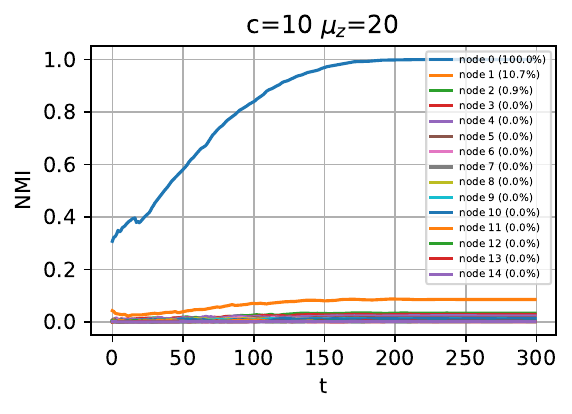}
        \\[-1ex]{\scriptsize (b) Maximum, complete ($d_i=n-1$)}
    \end{minipage}
    \begin{minipage}{0.24\textwidth}
        \centering
        \includegraphics[width=\linewidth]{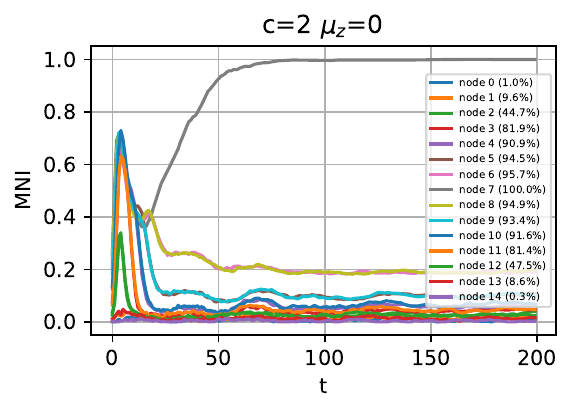}
        \\[-1ex]{{\footnotesize (c) Median, ring ($d_i=2$)}}
    \end{minipage}
    \hfill
    \begin{minipage}{0.24\textwidth}
        \centering
        \includegraphics[width=\linewidth]{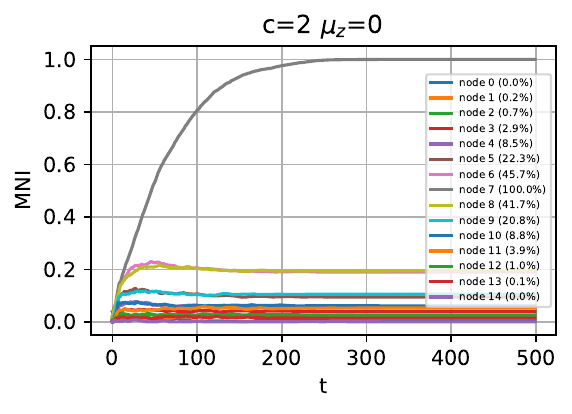}
        \\[-1ex]{\footnotesize (d) Median, complete ($d_i=n-1$)}
    \end{minipage}
    \caption{NMI $\frac{I(S_i;X_i^{(t)})}{I(S_i;S_i)}$ under ring topology and complete topology for both maximum and median consensus.}
    \label{fig:nmi_mediantopo}
\end{figure}

\subsection*{Part II. Comparison with the non-private counterpart and existing DP-based approaches}\label{ss:privacy_comp}

To further demonstrate the advantages of the proposed approach in terms of individual privacy, we compare it with the non-private counterpart and an existing DP-based approach. Fig.~\ref{fig:maximum_DP} reports the comparison for maximum consensus. The non-private counterpart exhibits substantial privacy leakage, with high NMI values for most nodes during the early iterations. The existing DP-based approach reduces the dependence between $X_i^{(t)}$ and $S_i$ by injecting noise at initialization; however, this privacy gain comes at the cost of persistent output error.

In contrast, the proposed approach effectively preserves the privacy of all non-target nodes, while still achieving exact aggregation accuracy (cf. Fig.~\ref{fig:nmi_maximum}(d)). Similar conclusions hold for median consensus, with corresponding results provided in Fig.~7 of the supplementary material.

\begin{figure}[ht]
    \centering
    \begin{minipage}{0.24\textwidth}
        \centering
        \includegraphics[width=\linewidth]{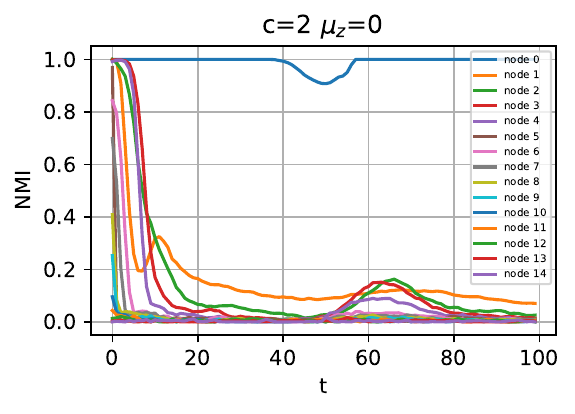}
        \\[-1ex]{\small (a) Non-private counterpart}
    \end{minipage}
    \hfill
    \begin{minipage}{0.24\textwidth}
        \centering
        \includegraphics[width=\linewidth]{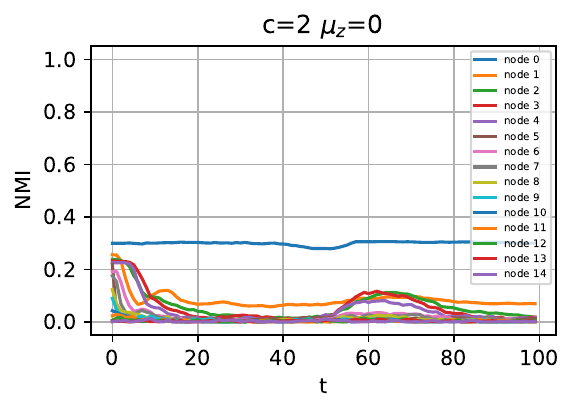}
        \\[-1ex]{\small (b) Existing DP-based approach }
    \end{minipage}
    \caption{Maximum consensus: NMI $\frac{I(S_i;X_i^{(t)})}{I(S_i;S_i)}$ for (a) the non-private counterpart and (b) the existing DP-based approach ($\sigma_s=0.1$).}
    \label{fig:maximum_DP}
\end{figure}

\section{Conclusion}\label{sec.conclusion}
In this work, we conducted a comprehensive information-theoretic privacy analysis of distributed nonlinear aggregation systems. Using mutual information, we established both fundamental privacy-leakage limits and the attainable leakage for various aggregation applications, showing that the ideal bound can be approached asymptotically under appropriate conditions. Simulation results validate the theory and illustrate how privacy depends on algorithmic and network parameters, offering practical insights for the design of privacy-preserving nonlinear aggregation schemes. 

\bibliographystyle{IEEEbib}
\bibliography{refs}

\appendices
\section{Proof of Proposition~\ref{prop:max_leak}}\label{app:proofs_max}
\begin{proof}
    Consider a distributed algorithm computing $s_{\max}=\max_{j\in\mathcal{V}} s_j$. Without loss of generality we can assume that the algorithm consists of $t_{max}$ rounds and that during each round $0\leq t\leq t_{max}$ a message is sent from node $i$ to node $j$ if $(i,j)\in \mathcal{E}$. We denote this message by $m_{i\to j}^{(t)}$ and remark that $m_{i\to j}^{(t)}$ can be an empty string if no message is sent from node $i$ to $j$ during round $t$. Node $i$ might need to choose some randomness during the computation and the random choices are collected in the vector $\mathbf{r}_i$. We remark, that any message sent from node $i$ in round $t$ can only depend on $s_i,\mathbf{r}_i$ and $m_{j\to i}^{\tau}$, $m_{i\to j}^{\tau}$ for $0\leq \tau<t$. With this notation note that an adversary eavesdropping every message and passively corrupting the nodes in $\mathcal{V}_c$ will have the knowledge of 
    \begin{align*}
        \mathcal{O}=\{\{S_i,R_i\}_{i\in \mathcal{V}_c},\{M_{i\to j}^{(t)}\}_{(i,j)\in \mathcal{E}},S_{max}\}
    \end{align*}
    After an execution of the algorithm, the adversary will possess an outcome of $\mathcal{O}$ given by
    \begin{align*}
        \{\{s_i,\mathbf{r}_i\}_{i\in \mathcal{V}_c},\{m_{i\to j}^{(t)}\}_{(i,j)\in \mathcal{E}},s_{max}\}.
    \end{align*}
    Remark that since all nodes learns the value $s_{max}$ each node $i$ should be able to compute this value from $s_i,\mathbf{r}_i$, and $\{m_{j\to i}^{(t)}\}_{j\in \mathcal{N}_i,t=0,1,\ldots,t_{max}}$. 
    To characterize which inputs of an honest node $k$ are consistent with the observed transcript, the adversary may perform a consistency test by re-running the algorithm with a trial value $\tilde{s}_k$ for node $k\in\mathcal{V}_h$. All messages of nodes other than $k$ are fixed to those in the original execution. The adversary then seeks a randomness realization $\tilde{\mathbf{r}}_k$ such that the messages generated by node $k$ match the observed ones.
Such a pair always exists at least for the true input, namely $\tilde{s}_k=s_k$ and $\tilde{\mathbf{r}}_k=\mathbf{r}_k$.\footnote{For privacy of node $k$, there should ideally exist multiple values of $\tilde{s}_k$ that pass this consistency test; otherwise the adversary can infer $s_k$.}

    If $s_k=s_{\max}$ (i.e., node $k$ attains the maximum), then any consistent re-execution must satisfy $\tilde{s}_k=s_k$. 
Indeed, if $\tilde{s}_k\neq s_k$ while all other inputs remain unchanged, the maximum of the input vector changes (it becomes strictly smaller if $\tilde{s}_k<s_k$ or strictly larger if $\tilde{s}_k>s_k$), so the correct output must change. But a consistent re-execution yields the same transcript $\mathsf{T}$ and therefore the same output, a contradiction. Hence $\tilde{s}_k=s_k$ is necessary. In contrast, if $s_k<s_{\max}$, then choosing any $\tilde{s}_k<s_{\max}$ does not change the true maximum, so consistency is not ruled out by correctness of the output (in particular, $\tilde{s}_k=s_k$ is always consistent).

Therefore, the unique node for which no alternative value $\tilde{s}_k\neq s_k$ can be consistent with the observed transcript is precisely the node attaining $s_{\max}$. Hence the adversary can identify the node having the maximal value.
\end{proof}

\section{Proof of Proposition~\ref{prop:med_leak}}\label{app:proofs_med}
\begin{proof}
    The arguments are very similar to the maximum case, where the adversary can re-run the algorithm re-using the messages for all except one node. It is not possible to get a successful re-run if $\tilde{s}_i$ is on the other side of the median than $s_i$. And for the node holding the median value it is not possible to choose any other $\tilde{s}_i$ than $s_i$. Hence, the adversary learns which node is holding the median value and for each node, it learns if $s_i$ is above or below the median value.  
\end{proof}

\section{Proof of Theorem~\ref{thm.max_prime}}\label{app:max_prime}

\begin{proof}

For node $i\in\mathcal{V}_p$, we have
\begin{align} \label{eq.x_diff}
     x_j^{(t+1)}-x_j^{(t)} &=\frac{-\sum_{  k \in {\mathcal N}_j} a_{  j k} (z_{  j|k}^{(t+1)}-  z_{  j|k}^{(t)})}{cd_j}.
\end{align}
Also from \eqref{eq.zji_up}, $z_{j\mid i}^{(t+1)}-z_{j\mid i}^{(t)}$ can be rewritten as 
\begin{align} 
  \nonumber z_{j\mid i}^{(t+1)}-z_{j\mid i}^{(t)}&=ca_{ij}x_i^{(t)}-\frac{1}{2}ca_{ij}x_i^{(t-1)}+\frac{1}{2}ca_{ji}x_j^{(t-1)}
 \\
&= ca_{ij} \left(x_i^{(t)} - \frac{1}{2}x_i^{(t-1)} - \frac{1}{2}x_j^{(t-1)}\right).
\label{eq.z_diff} 
\end{align}
As a result $x_j^{(t+1)}$ can be computed from $ x_j^{(t)}$ and $ x_j^{(t-1)}$ and $ x_k^{(t)}$ and $ x_k^{(t-1)}$ for neighboring nodes $k$.

Firstly, we can separate $I(S_i;\mathcal{O})$ as
\begin{align*}
    I(S_i;\mathcal{O})\overset{(a)}{=}&I(S_i;\mathcal{O},Y_i)\overset{(b)}{=}I(S_i;Y_i)+I(S_i;\mathcal{O}|Y_i)\\
    \overset{(c)}{=}&I(S_i;Y_i)+P(Y_i=0)I(S_i;\mathcal{O}|Y_i=0)\\
    &+P(Y_i=1)I(S_i;\mathcal{O}|Y_i=1)
\end{align*}
where (a) follows from the fact that $Y_i$ can be inferred by observing whether 
$X^{(t)}$ is computed from $X^{(t-1)}$ and $X^{(t-2)}$ by \eqref{eq.x_diff} and \eqref{eq.z_diff}; 
(b) follows from the chain rule of mutual information, 
$I(A;B,C)= I(A;B \mid C) + I(A;C)$; 
and (c) follows from the law of total expectation.

Let $i \in \mathcal{V} \setminus \mathcal{V}_P$. According to \eqref{eq.xi_up}, there exists some iteration $t \in \mathcal{T}$ such that $x_i^{(t)} = s_i$. Therefore, we have
\begin{align*}
    I(S_i;\mathcal{O}|Y_i=0)=&I(S_i;S_i\mid Y_i=0)
\end{align*}

For node $i\in\mathcal{V}_p$, we can obtain
\begin{align*}
    &I(S_i;\mathcal{O}|Y_i=1)\\
    \overset{(a)}{=}&I(S_i;\{S_j\}_{j\in\mathcal{V}_c},\{X_{j}^{(t)}\}_{j\in \mathcal{V},t\in\mathcal{T}},\{Z_{j\mid k}^{(0)}\}_{(j,k)\in\mathcal{E}}\mid Y_i=1)\\
    \overset{(b)}{=}&I(S_i;\{S_j\}_{j\in\mathcal{V}_c},\{Y_j\}_{j\in \mathcal{V}_p},\{Y_j,S_j\}_{j\in \mathcal{V}\backslash \mathcal{V}_p},\\
    &\{Z_{j\mid k}^{(0)}\}_{(j,k)\in\mathcal{E}}\mid Y_i=1)\\
    \overset{(c)}{=}&I(S_i;\{S_j\}_{j\in\mathcal{V}_c}\cup\{S_j\}_{j\in \mathcal{V}\backslash \mathcal{V}_p},\{Y_j\}_{j\in \mathcal{V}}\mid Y_i=1)\\
\end{align*}
where $(a)$ is by definition and \eqref{eq.z_diff}. $(b)$ holds because, for any node $j \in \mathcal{V}_p$, the variable $x_j^{(t)}$ is determined by the previous states $x_k^{(t-1)}$ and $x_k^{(t-2)}$ within its neighborhood $k \in \mathcal{N}_j \cup \{j\}$. By recursive application, $x_j^{(t)}$ can be fully reconstructed from the initial $z_{j\mid k}^{(0)}$. Consequently, the adversary can deduce the indicator variables $\{Y_j\}_{j\in \mathcal{V}}$ and the subset of private values $\{S_j\}_{j\in \mathcal{V}_p}$. $(c)$ comes from the independency of variable $Z$.

\end{proof}

\section{Proof of Lemma~\ref{lm.1}}\label{app:lm}
\begin{proof}
    For any honest node $i \in \mathcal{V}_h$ with $s_i \neq  s_{\max}$, we obtain
    \begin{align*}
&I(S_i;\mathcal{O})\\
&\overset{(a)}{=}I(S_i;Y_i)+I(S_i;\mathcal{O}|Y_i=1)\\
&\overset{(b)}{=}I(S_i;Y_i)+I(S_i;\{S_j\}_{j\in\mathcal{V}_c}\cup\{S_j\}_{j\in \mathcal{V}\backslash \mathcal{V}_p},\{Y_j\}_{j\in \mathcal{V}}\mid Y_i=1)\\
&\overset{(c)}{=}I(S_i;Y_i)+I(S_i;\{S_j\}_{j\in\mathcal{V}_c}\cup\{S_j\}_{j\in \mathcal{V}\backslash \mathcal{V}_p},\{Y_j\}_{j\in \mathcal{V}}\mid Y_i)\\
&\overset{(d)}{=}I(S_i;\{S_j\}_{j\in\mathcal{V}_c}\cup\{S_j\}_{j\in \mathcal{V}\backslash \mathcal{V}_p},\{Y_j\}_{j\in \mathcal{V}})\\
&\overset{(e)}{=}I(S_i;\{S_j\}_{j\in\mathcal{V}_c},S_{\max},\{Y_j\}_{j\in \mathcal{V}})\\
&\overset{(f)}{=}I(S_i;\{S_j\}_{j\in\mathcal{V}_c},S_{\max},\{\mathbb{I}_{S_j=S_{max}}(S_j)\}_{j\in \mathcal{V}})
\end{align*}
where $(a)$ comes from \eqref{eq.theo11} and $P(Y_i=1) = 1$; in $(b)$ we plug in \eqref{eq.theo13}. $(c)$ comes from the fact when $P(C=1) = 1$ the term $I(A;B|C=1)$ and $I(A;B|C)$ are equivalent. $(d)$ comes from $I(A;B,C)=I(A;C)+I(A;B|C)$. $(e)$ comes from the case when all nodes $j$ with $s_j \neq s_{\max}$ are in $\mathcal{V}_p$. $(f)$ arises from $c \gg 1$, , where the step  becomes infinitesimally small. With a proper initialization, the lowest point along the trajectory of $x_i$ asymptotically approaches $s_{\max}$.
\end{proof}

\clearpage

\section*{Supplementary material}
\subsection{Additional figures}


\begin{figure}[ht]
    \centering
    \begin{minipage}{0.24\textwidth}
        \centering
        \includegraphics[width=\linewidth]{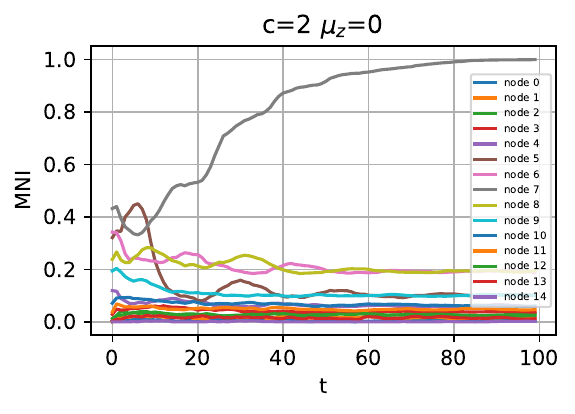}
        \\[-1ex]
        {\small (a) Non-private counterpart}
    \end{minipage}
    \hfill
    \begin{minipage}{0.24\textwidth}
        \centering
        \includegraphics[width=\linewidth]{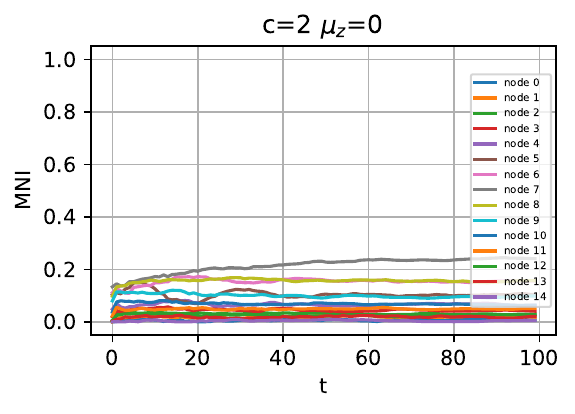}
        \\[-1ex]
        {\small (b) Existing DP-based approach}
    \end{minipage}

    \caption{Median consensus: NMI $\frac{I(S_i;X_i^{(t)})}{I(S_i;S_i)}$ for (a) the non-private counterpart and (b) the existing DP-based approach ($\sigma_s=0.1$).}
    \label{fig:median_DP}
\end{figure}

\subsection{DP with Laplacian noise}
Fig.~\ref{fig:mse_dp} shows the output accuracy under the DP-based approaches with Gaussian and Laplacian noise perturbation. 
\begin{figure}[ht]
    \centering
    \begin{minipage}{0.24\textwidth}
        \centering
        \includegraphics[width=\linewidth]{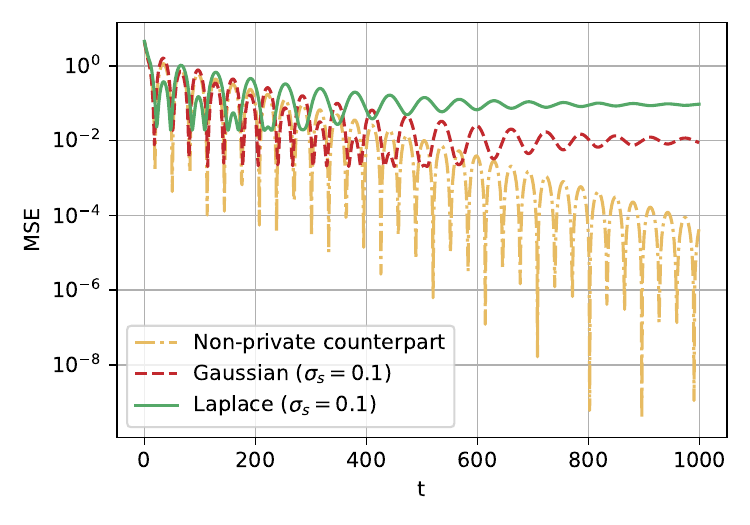}
        \\[-1ex] 
        {\small (a) Maximum consensus}
    \end{minipage}
    \hfill
    \begin{minipage}{0.24\textwidth}
        \centering
        \includegraphics[width=\linewidth]{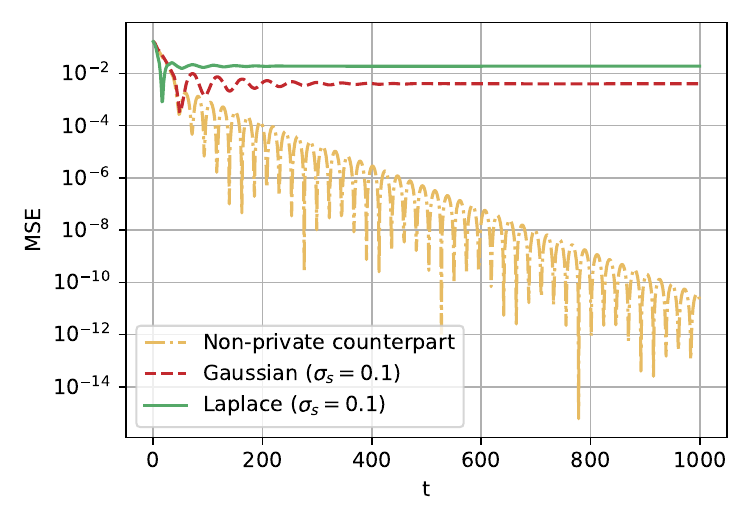}
        \\[-1ex]
        {\small (b) Median consensus}
    \end{minipage}
    \caption{Output accuracy (MSE) of maximum and median consensus under the non-private counterpart and the DP-based approaches with Gaussian and Laplacian noise perturbation.}
    \label{fig:mse_dp}
\end{figure}

Fig.~\ref{fig:nmi_dp} shows the Laplacian-noise variant of the DP approach, under the same settings as in Figs.~\ref{fig:maximum_DP}(b) and~\ref{fig:median_DP}(b).
\begin{figure}[ht]
    \centering
    \begin{minipage}{0.24\textwidth}
        \centering
        \includegraphics[width=\linewidth]{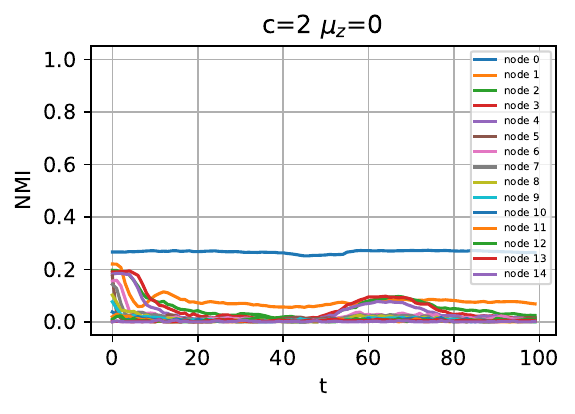}
        \\[-1ex]{\small (a) Maximum consensus with Laplacian perturbation ($\sigma_s=0.1$)}
    \end{minipage}
    \hfill
    \begin{minipage}{0.24\textwidth}
        \centering
        \includegraphics[width=\linewidth]{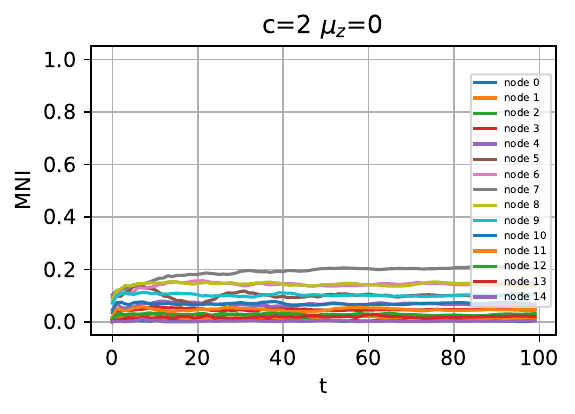}
        \\[-1ex]{\small (b) Median consensus with Laplacian perturbation ($\sigma_s=0.1$)}
    \end{minipage}
    \caption{NMI $\frac{I(S_i;X_i^{(t)})}{I(S_i;S_i)}$ of Laplacian perturbation for (a) maximum consensus and (b) median consensus.}
    \label{fig:nmi_dp}
\end{figure}

\subsection{Other results}

As shown in Fig.~\ref{fig:min}, we present the experimental results for the minimum consensus case. In contrast to the maximum consensus, the soft adjustment of $z^{(0)}$ drives the trajectory of $x$ to gradually increase from a point below all private values and eventually converge to the minimum value. The behavior of NMI also exhibits a distinct pattern, which only the node associated with the minimum private value sees its NMI approach $1$ upon convergence, while all other nodes remains low.
\begin{figure}[ht]
    \centering
    \begin{minipage}{0.24\textwidth}
        \centering
        \includegraphics[width=\linewidth]{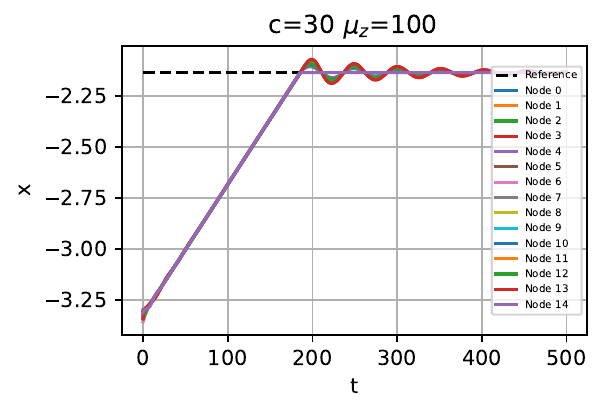}
        \\[-1ex]
        {\small (a)}
    \end{minipage}
    \hfill
    \begin{minipage}{0.24\textwidth}
        \centering
        \includegraphics[width=\linewidth]{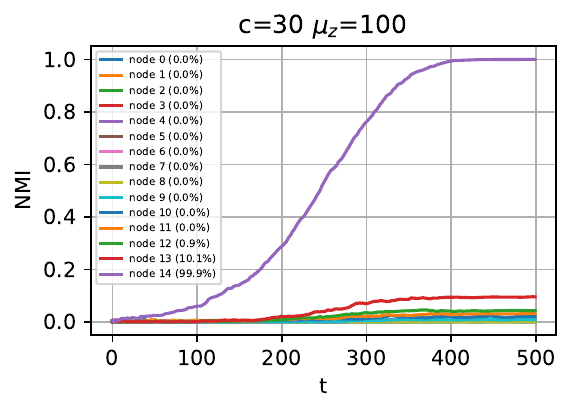}
        \\[-1ex]
        {\small (b)}
    \end{minipage}
    \caption{Trajectory and NMI
    $\frac{I(S_i;X_i^{(t)})}{I(S_i;S_i)}$ of minimum consensus.}
    \label{fig:min}
\end{figure}

Fig.~\ref{fig:topk} illustrates the process of obtaining the top-
$K$ values within a distributed network. It can be seen that this process is similar to the maximum consensus, sequentially identifying the largest value, the second largest, and so on. Apart from the top-$K$ nodes, the NMI of all other nodes remains close to 0. Since the iteration is performed for a finite number of rounds, we set $\epsilon=1e-3$ as the tolerance for determining convergence. This may cause slight misclassification in some experiments, leading to the mutual information of the second-largest node being lower than $1$.
\begin{figure}[ht]
    \centering
    \begin{minipage}{0.24\textwidth}
        \centering
        \includegraphics[width=\linewidth]{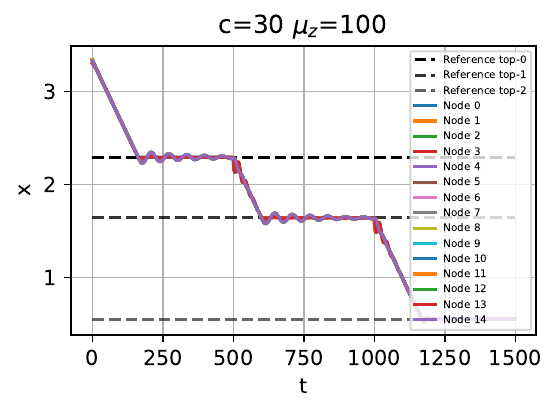}
        \\[-1ex]
        {\small (a)}
    \end{minipage}
    \hfill
    \begin{minipage}{0.24\textwidth}
        \centering
        \includegraphics[width=\linewidth]{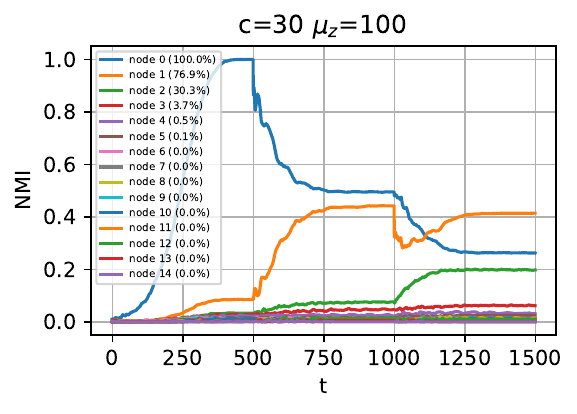}
        \\[-1ex]
        {\small (b)}
    \end{minipage}

    \caption{Trajectory and NMI
    $\frac{I(S_i;X_i^{(t)})}{I(S_i;S_i)}$ of top-$K$ aggregation ($K=3$).}
    \label{fig:topk}
\end{figure}

Fig.~\ref{fig:quantile} shows the quantile consensus process (for the 5th-ranked value). By initializing $x$ from a relatively large value, the trajectory converges downward, which implies a higher risk of information leakage for larger values.
\begin{figure}[ht]
    \centering
    \begin{minipage}{0.24\textwidth}
        \centering
        \includegraphics[width=\linewidth]{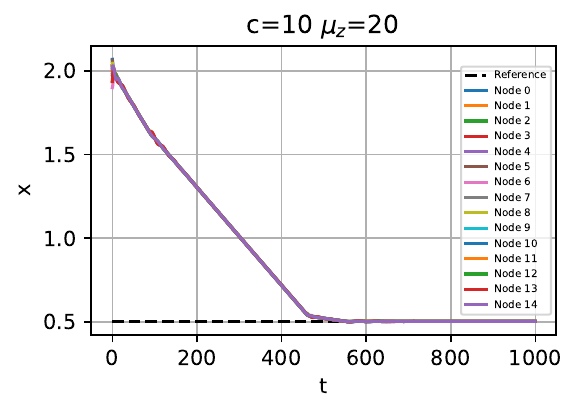}
        \\[-1ex]
        {\small (a)}
    \end{minipage}
    \hfill
    \begin{minipage}{0.24\textwidth}
        \centering
        \includegraphics[width=\linewidth]{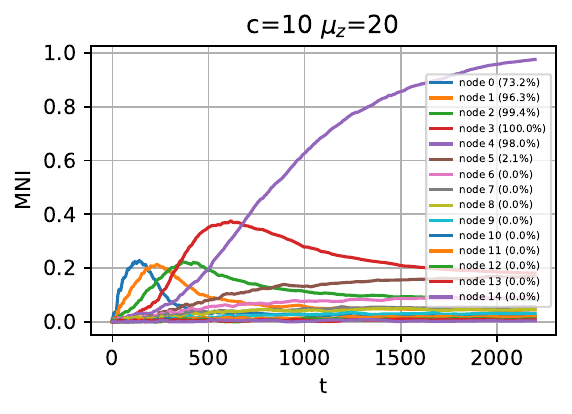}
        \\[-1ex]
        {\small (b)}
    \end{minipage}

    \caption{Trajectory and NMI
    $\frac{I(S_i;X_i^{(t)})}{I(S_i;S_i)}$ of quantile consensus (the $5$th-ranked value with $n=15$, i.e. $q=\frac{2}{7}$).}
    \label{fig:quantile}
\end{figure}

Fig.~\ref{fig:trim} illustrates the process of obtaining the trimmed values after removing the extreme values at both ends. With appropriate initialization, only the extreme values reveal their private information, while the private values of the remaining nodes cannot be directly inferred from $x$.
\begin{figure}[ht]
    \centering
    \begin{minipage}{0.24\textwidth}
        \centering
        \includegraphics[width=\linewidth]{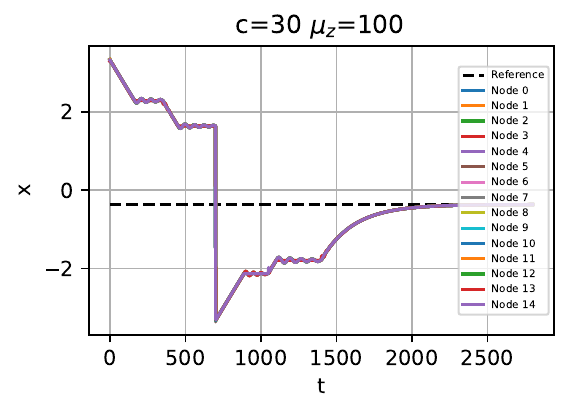}
        \\[-1ex]
        {\small (a)}
    \end{minipage}
    \hfill
    \begin{minipage}{0.24\textwidth}
        \centering
        \includegraphics[width=\linewidth]{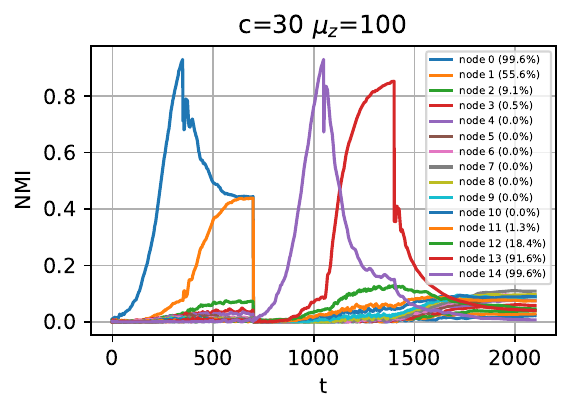}
        \\[-1ex]
        {\small (b)}
    \end{minipage}

    \caption{Trajectory and NMI
    $\frac{I(S_i;X_i^{(t)})}{I(S_i;S_i)}$ of trimmed mean consensus. Mean value computed after excluding the 2 largest and 2 smallest values.}
    \label{fig:trim}
\end{figure}

\end{document}